\documentclass[12pt]{article}

\usepackage[letterpaper,top=2cm,bottom=2cm,left=3cm,right=3cm,marginparwidth=1.75cm]{geometry}

\usepackage{amsmath,amssymb, amsthm}
\usepackage{graphicx}
\usepackage{ytableau}
\usepackage{booktabs}
\usepackage[colorlinks=true, allcolors=blue]{hyperref}
\usepackage{enumerate}
\usepackage[shortlabels]{enumitem}

\usepackage[numbers,sort&compress]{natbib}

\usepackage{tikz,fullpage}
\usetikzlibrary{arrows,
        petri,
        topaths,
        positioning,
        shapes.geometric,
        calc,arrows.meta}
\usepackage{tkz-berge}

\usepackage{pgf}

\usetikzlibrary{arrows,automata}
\usepackage[utf8]{inputenc}

\usepackage{caption}
\usepackage{subcaption}
\usepackage{xcolor}

\usepackage{algorithm}
\usepackage{algpseudocode}

\usepackage{authblk}

\newtheorem{Theorem}{Theorem}[section]

\newtheorem{Lemma}[Theorem]{Lemma}
\newtheorem{Corollary}[Theorem]{Corollary}

\newtheorem{Problem}[Theorem]{Problem}

\theoremstyle{definition}
\newtheorem{Definition}[Theorem]{Definition}
\newtheorem{Observation}[Theorem]{Observation}
\newtheorem{Remark}[Theorem]{Remark}

\DeclareMathOperator{\supp}{supp}

\title{Tree Buckets and the Reconstruction of Pairs of Phylogenetic Trees}
\author[1]{Sky Basire}
\author[1, *]{Michael Hendriksen}
\affil[1]{Department of Mathematics and Statistics, University of New South Wales}
\affil[*]{Corresponding author: Michael Hendriksen, m.hendriksen@unsw.edu.au}

\begin{document}

\maketitle

\begin{abstract}
Phylogenetic trees are used in evolutionary biology to represent the evolutionary history of a collection of taxa. As we have incomplete information about any evolutionary history, recovering trees from partial information is a focus of phylogenetic combinatorics. However, in some cases the available data does not describe a single phylogenetic tree. We consider recovery of pairs of phylogenetic trees from their combined subtrees with $k$ leaves, which we call a $k$-bucket. We establish the exact cases in which these pairs of trees are recoverable from their subtrees with a single leaf removed, both when just considering the structure of the trees, and when additionally considering the set of taxa on the leaves. We also consider recovery of pairs of trees with labelled leaves from their rooted triples, and establish that they are recoverable up to a sequence of subtree swaps.

\end{abstract}

\section{Introduction}

Reconstruction of evolutionary histories in phylogenetics often requires combining partial information from multiple sources \cite{kapli2020phylogenetic} \cite{delsuc2005phylogenomics}. For instance, when reconstructing a phylogenetic tree, it is well-known that a single labelled tree can be reconstructed if each of its labelled $3$-leaf subtrees are known \cite{aho1981inferring}, and a single unlabelled tree can be reconstructed if each of its $(n-1)$-leaf subtrees are known \cite{Decks_of_rooted_binary_trees}. Recent work has considered the general issue of reconstructing a phylogenetic tree from partial information, for instance if only some of the subtrees with a particular number of leaves are known \cite{Decks_of_rooted_binary_trees}. However, numerous reconstruction techniques can result in conflicting subtrees, e.g. sampling of different genes under the presence of incomplete lineage sorting or from histories with horizontal gene transfer or other recombination events \cite{felsenstein2003inferring}. Furthermore, in settings such as metagenomic analyses \cite{delsuc2005phylogenomics} or supertree reconstruction \cite{bininda2002super}, one may have to reconcile conflicting partial information arising from different reconstruction techniques.

A natural question, therefore, is how to determine whether, given partial information, say the set of subtrees with a given number of leaves, some of which may conflict with each other, is consistent with a set of multiple trees.

In the present manuscript, we develop the systematic study of this problem. In particular, we consider the problem of whether, given a set of subtrees that may disagree for some or all of their structure, when one can tell whether these subtrees are consistent with a set of two trees. We model this by taking a pair of trees, combining their sets of $k$-leaf subtrees into one multiset, which we call the $k$-bucket, and asking when those original two trees can be recovered.

In contrast to the unique reconstruction possible in the classical case of a single tree, in the case of mixed signals from two trees various phenomena can arise that make it impossible to tell from which initial tree a signal has arisen. For instance, mixed signals from pairs of labelled trees may admit non-trivial transformations that preserve their $k$-bucket. In particular, a swap of a suitable subtree between the two may result in an identical multiset of $k$-leaf subtrees.

Similar problems also have a strong historical context in combinatorics, and in particular the present problem may be viewed as a natural analogue of the well-known graph reconstruction conjecture \cite{graph_reconstruction_survey}, which hypothesises that a graph can be reconstructed from its induced subgraphs. The primary differences between our problem and classical graph reconstruction is twofold; firstly we consider partial information generated by two objects rather than one, and secondly, the information is pooled, so further difficulty is imposed by the loss of provenance of the partial objects.

In the present paper, we obtain three main results in this context. We show that for a pair of unlabelled trees, aside from a finite number of small pathological cases with $8$ or fewer leaves (which we characterise), one can always reconstruct the original trees. We further show that reconstruction can be impossible for arbitrarily large pairs of trees from their rooted triples, but characterise the phenomenon from which this arises - the aforementioned so-called subtree swap. Finally, we show that, again, aside from some small pathological cases with six or fewer leaves (again characterised), reconstruction is always possible for labelled trees from their $(n-1)$-leaf subtrees.

We anticipate that the present results will provide a foundation for future extensions of this problem, such as mixed signals from a larger number of trees, or pairs of trees with only a partial $k$-bucket provided. In particular, our present results provide extremal cases, as we assume that all subtrees with $k$ leaves are known for both trees, and still there exist cases of pairs of trees that cannot be recovered. This shows that if one does not know from which tree each obtained subtree has come from, even with the maximum possible information in terms of subtree structure it may be impossible to recover the pair of input trees. 

In Section \ref{s:prelim} we establish the background theory of phylogenetic trees and formally introduce the concept of $k$-buckets. In Section \ref{s:general} we define reconstructibility of pairs of trees, and consider reconstruction of a number of tree statistics - in particular, root balance, which forms a key part of later proofs. In Section \ref{s:unlabelledn-1}, we show that reconstruction of pairs of unlabelled trees from their $(n-1)$-buckets is always possible, aside from some small pathological counterexamples. In Section \ref{s:label3}, we characterise reconstruction of pairs of labelled trees from their $3$-bucket, and classify exactly when reconstruction is impossible (although this can occur for arbitrarily large trees). In Section \ref{s:labeln-1}, we then show that reconstruction of pairs of labelled trees from their $(n-1)$-buckets is always possible, again aside from some small pathological counterexamples. We then conclude in Section \ref{s:discuss} with some potential pathways for future research.

\section{Preliminaries}\label{s:prelim}

In the present paper we exclusively consider rooted, binary phylogenetic trees, which are permitted to be either labelled or unlabelled. We will generally omit some or all of the adjectives ``rooted, binary, phylogenetic'', but all trees considered are of this form. However, we do consider both labelled and unlabelled trees and this will always be specified.

\begin{Definition}
  A \emph{(rooted) binary phylogenetic tree} $T$ on $n$ leaves is a directed acyclic graph whose vertex set satisfies the following:
  \begin{itemize}
    \item There is one vertex of in-degree 0 and out-degree $2$, called the \emph{root};
    \item There are $n$ vertices of in-degree 1 and out-degree 0, called \emph{leaves}; and
    \item All other vertices, the \emph{internal} vertices, have in-degree  $1$ and out-degree $2$.
  \end{itemize}
  If, furthermore, there exists a set $X$ such that $|X|=n$ and the leaves of $T$ are bijectively labelled by the elements of $X$, then $T$ is referred to as a \emph{labelled} phylogenetic tree. Denote the set of labelled binary phylogenetic trees on $X$ by $T(X)$, and the set of unlabelled binary phylogenetic trees on $n$ leaves by $UT(n)$.
\end{Definition}

We will require some standard definitions used to describe trees as well.

\begin{Definition}
    Let $T$ be a phylogenetic tree and $u,v \in V(T)$. If there is an arc $(u,v) \in E(T)$, then $u$ is referred to as the parent of $v$, and $v$ is referred to as a child of $u$. If there is some pair of leaves $\ell_1,\ell_2$ such that there is a third distinct vertex $v$ that is the parent of both $\ell_1$ and $\ell_2$, then $\ell_1$ and $\ell_2$ are referred to as a \emph{cherry}. A phylogenetic tree that contains precisely one cherry is referred to as a \emph{caterpillar} tree. Suppose a caterpillar with $k$ leaves has cherry $(a_1,a_2)$ with parent $v$, and the leaf child of the parent of $v$ is $a_3$, the leaf child of the parent of the parent of $v$ is $a_4$, and so on. We will denote such a caterpillar by $a_1a_2a_3 \dots a_k$ (noting this is unique up to the order of $a_1$ and $a_2$).
\end{Definition}

 As an example of the above notation for caterpillar trees, $T_2$ in Figure \ref{fig:bucket_example} would be denoted $abecd$ (or $baecd$).

\begin{Definition}
    Let $T$ be a phylogenetic tree, and $P=\big((v_0,v_1),(v_1,v_2)\dots,(v_{k-1},v_k) \big)$ a sequence of distinct edges in $T$. Then $P$ is referred to as a \emph{path}. The number of edges $k$ in $P$, is referred to as the \emph{length} of $P$. The \emph{height} of a tree $T$ is the length of the longest path in $T$ such that $v_0$ is the root of $T$ and $v_k$ is a leaf. 
\end{Definition}

For example, $T_2$ in Figure \ref{fig:bucket_example} has a height of $4$, given by the path from the root to leaf $a$.

\begin{Definition}
    Let $T$ be a phylogenetic tree and $v$ a degree-$2$ vertex in $V(T)$, with parent $u$ and child $w$. Then the operation of deleting $v$ and its adjacent edges, and adding the edge $(u,w)$ is referred to as \emph{suppressing} $v$.
\end{Definition}

\begin{Definition}
    Let $T$ be a labelled phylogenetic tree and $Y \subseteq X$. Then the \emph{subtree of $T$ induced by $Y$} is the minimal connected subgraph of $T$ that contains the leaves labelled by $Y$, with any vertices of both in-degree and out-degree $1$ suppressed. If $S$ is a subtree of $T$ induced by $Y$, we may simply refer to $S$ as a subtree of $T$ without specification of $Y$.
\end{Definition}

Recent research has considered the \emph{$k$-multideck} of a given tree $T$, which is the multiset of subtrees of $T$ with $k$ leaves. In particular, Clifton et al. \cite{Decks_of_rooted_binary_trees} consider problems regarding their cardinality and extremal values, as well as similar problems for a related concept, the \emph{$k$-deck}, which ignores multiplicity and will not be considered in this manuscript.

\begin{Definition}
    Let $T$ be a phylogenetic tree. If $T$ is labelled, then the \emph{$k$-multideck} of $T$, denoted $M_k(T)$ is the multiset of subtrees of $T$ induced by $Y \subseteq X$ where $|Y|=k$. If $T$ is not labelled, then the \emph{$k$-multideck} of $T$ is defined to be the multiset of minimal connected subgraphs on $T$ that contain $k$ leaves, where those $k$ leaves are leaves in $T$, and vertices of both in-degree and out-degree $1$ are suppressed.
\end{Definition}

The multideck is defined for a single phylogenetic tree, but in the present manuscript we wish to consider pairs of trees.

\begin{Definition}
    Let $T_1$ and $T_2$ be two phylogenetic trees and define $B_k(\{T_1,T_2\}) = M_k(T_1) \cup M_k(T_2)$ (defined as a multiset). We refer to this as the \emph{$k$-bucket} of $T_1$ and $T_2$.
\end{Definition}

Motivated by recoverability of pairs of phylogenetic trees, we have the following natural question.

\begin{Problem}
    Let $T_1$ and $T_2$ be a pair of phylogenetic trees. When does $B_k(\{T_1,T_2\})$ have a unique pre-image?
\end{Problem}
We note in particular that we consider $T_1$ and $T_2$ as a set, so order does not matter in their recovery. Furthermore, although we define $k$-buckets generally, we will be focusing on the $k=3$ and $k=n-1$ cases.

\begin{Observation}
    The function $B_k$ never has a unique pre-image for unlabelled trees if $k=3$, as there is a unique unlabelled tree on $3$ leaves. Thus, if $T_1,T_2$ are any pair of unlabelled trees on $n$ leaves, $B_k(\{T_1,T_2\})$ is simply the multiset of $\binom{n}{3}$ copies of the unlabelled $3$-leaf tree.
\end{Observation}

Due to this, while we consider $(n-1)$-buckets for both labelled and unlabelled trees, we will only consider $3$-buckets of labelled trees. Finally, we note that throughout this manuscript we assume that a $k$-bucket has been obtained from \emph{some} pair of trees, and is not just an arbitrary multiset of trees on $k$ leaves. We leave the full characterisation of whether a multiset is such a $k$-bucket to future research.

\begin{figure}[ht]
    \centering
    \begin{subfigure}[b]{0.45\textwidth}
        \centering
        \caption*{$T_1$}
        \includegraphics[width=\textwidth]{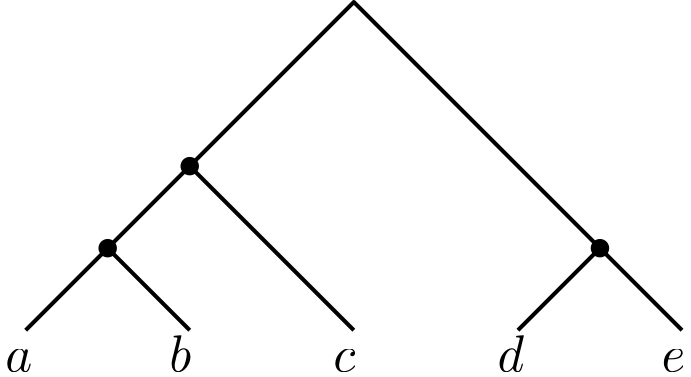}
    \end{subfigure}
    \hfill
    \begin{subfigure}[b]{0.45\textwidth}
        \centering
        \caption*{$T_2$}
        \includegraphics[width=\textwidth]{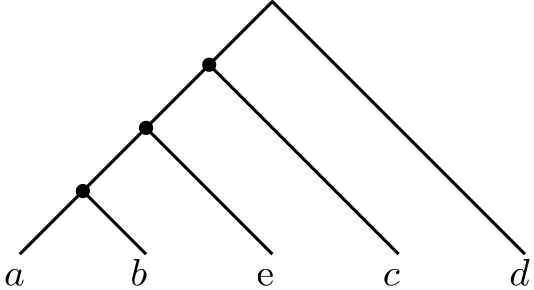}
    \end{subfigure}
    \newline
    \newline
    \begin{subfigure}[b]{0.18\textwidth}
        \centering
        \includegraphics[width=\textwidth]{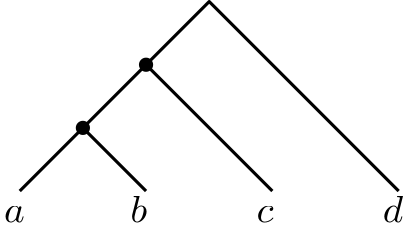}
    \end{subfigure}
    \hfill
    \begin{subfigure}[b]{0.18\textwidth}
        \centering
        \includegraphics[width=\textwidth]{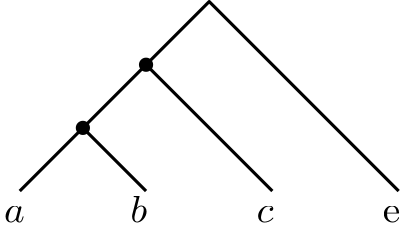}
    \end{subfigure}
    \hfill
    \begin{subfigure}[b]{0.18\textwidth}
        \centering
        \includegraphics[width=\textwidth]{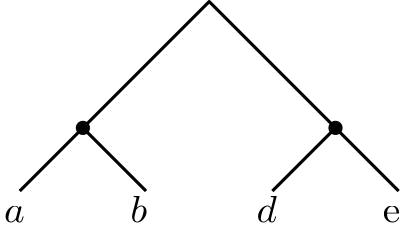}
    \end{subfigure}
    \hfill
    \begin{subfigure}[b]{0.18\textwidth}
        \centering
        \includegraphics[width=\textwidth]{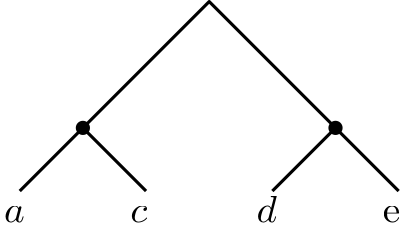}
    \end{subfigure}
    \hfill
    \begin{subfigure}[b]{0.18\textwidth}
        \centering
        \includegraphics[width=\textwidth]{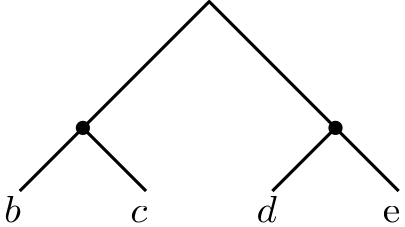}
    \end{subfigure}
    \newline
    \newline
    \begin{subfigure}[b]{0.18\textwidth}
        \centering
        \includegraphics[width=\textwidth]{bucket_example_subtree_1_1.png}
    \end{subfigure}
    \hfill
    \begin{subfigure}[b]{0.18\textwidth}
        \centering
        \includegraphics[width=\textwidth]{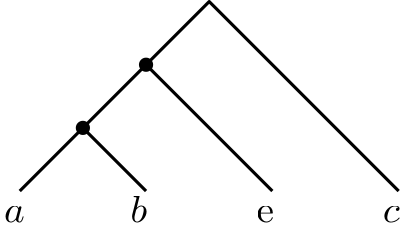}
    \end{subfigure}
    \hfill
    \begin{subfigure}[b]{0.18\textwidth}
        \centering
        \includegraphics[width=\textwidth]{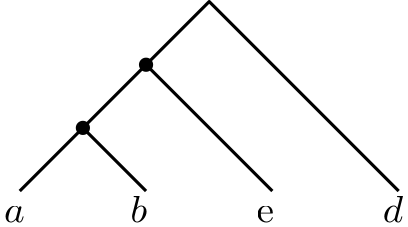}
    \end{subfigure}
    \hfill
    \begin{subfigure}[b]{0.18\textwidth}
        \centering
        \includegraphics[width=\textwidth]{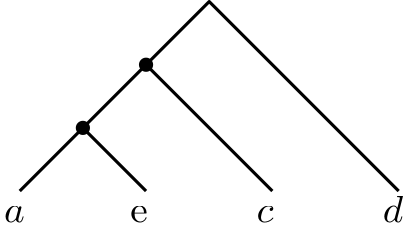}
    \end{subfigure}
    \hfill
    \begin{subfigure}[b]{0.18\textwidth}
        \centering
        \includegraphics[width=\textwidth]{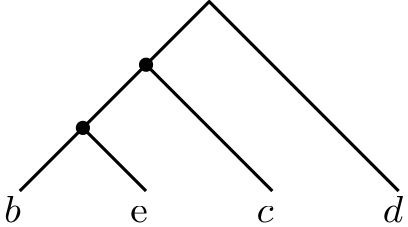}
    \end{subfigure}
    \caption{Labelled $5$ leaf phylogenetic trees $T_1$ and $T_2$ (above), and their $4$-bucket $B_4(\{T_1, T_2\})$ (below). The first row of subtrees form the $4$-multideck of $T_1$, and the second row form the $4$-multideck of $T_2$. Note that the first subtree of each row, $abcd$, is the same, as it is a subtree of both $T_1$ and $T_2$, and thus has multiplicity $2$ in their $4$-bucket.}
    \label{fig:bucket_example}
\end{figure}

We are now in a position in which we can begin to prove a number of theorems for recoverability.

\section{General results}\label{s:general}

We will first formalise the definition of what it means to recover a pair of trees when presented with their $k$-bucket.

\begin{Definition}
        Let $T_1$ and $T_2$ be rooted phylogenetic trees on $n$ leaves, and fix $0<k\le n$. If, for any pair of rooted phylogenetic trees $T_3$ and $T_4$ such that $B_k(\{T_1, T_2\}) = B_k(\{T_3, T_4\})$, then $\{T_1,T_2\}=\{T_3,T_4\}$, we say that $T_1$ and $T_2$ are \emph{recoverable from their $k$-bucket}. Similarly, if $B_k(\{T_1, T_2\}) = B_k(\{T_3, T_4\})$ implies that $T_1 \in \{T_3,T_4\}$, we say that $T_1$ is recoverable from their $k$-bucket.
\end{Definition}

Of course, it may be the case that full recoverability of a pair of trees is not possible, but some structural information about the input trees can still be deduced - for instance, the number of cherries in the input trees, or their heights. We define this possibility in a general way.

\begin{Definition}
    Let $A \in \{T(X),UT(n)\}$ and $f: A \rightarrow S$ be a function. Let $T_1, T_2 \in A$. We say $f$ is \emph{recoverable} from $B_k(\{T_1,T_2\})$ if, for any other pair $T_3,T_4 \in A$ such that $B_k(\{T_3,T_4\})=B_k(\{T_1,T_2\})$, we have that $\{f(T_1),f(T_2)\}= \{f(T_3),f(T_4)\}$.  
\end{Definition}

In the above definition, we intend to consider $f$ as a function that returns some structural information about the tree. For instance, if $c(T)$ is a function that returns the number of cherries in $T$, if $T_1$ has four cherries, $T_2$ has two cherries, and $c$ is recoverable for the $k$-bucket of $T_1,T_2$, then any pair of trees that has the same $k$-bucket as $T_1$ and $T_2$ would also have two and four cherries in some order.

Note that if $f$ is the identity on $A$, this is equivalent to the previous definition of recoverability for $\{T_1,T_2\}$. 

\subsection{General results for \texorpdfstring{$k$}{k}-bucket recoverability}

We first show that a $k$-bucket consisting of larger trees is at least as informative as a $k$-bucket consisting of smaller trees.

\begin{Theorem}\label{smaller_bucket_from_larger}
    Let $T_1$ and $T_2$ be either labelled or unlabelled rooted phylogenetic trees. Then $B_j(\{T_1,T_2\})$ is uniquely determined by $B_k(\{T_1,T_2\})$ if $j \leq k$.
\end{Theorem}

\begin{proof}
    Consider some tree $T \in B_j(\{T_1,T_2\})$ on a set of leaves $Y \subseteq X$. This tree is a subtree of some number of trees in $B_k(\{T_1,T_2\})$. In particular, any such tree will contain each leaf in $Y$, as well as $k-j$ other leaves. Therefore there are $\binom{n-j}{k-j}$ trees in $B_k(\{T_1,T_2\})$ with $T$ as a subtree. Consequently, if we consider the multiset of subtrees with $j$ leaves of trees in $B_k(\{T_1,T_2\})$, we will obtain $B_j(\{T_1,T_2\})$ with each occurring with a multiplicity of $\binom{n-j}{k-j}$. Therefore by dividing each multiplicity by this factor we can obtain $B_j(\{T_1,T_2\})$ from $B_k(\{T_1,T_2\})$, and so $B_j$ is uniquely determined by $B_k(\{T_1,T_2\})$.
\end{proof}

The following theorem is well-known, due to \cite{aho1981inferring}.

\begin{Theorem}[\cite{aho1981inferring}]\label{t:3-recoverable}
    Let $T$ be a labelled tree. Then $T$ has a unique $3$-multideck.
\end{Theorem}

Inspired by Theorem \ref{smaller_bucket_from_larger}, we can immediately generalise this theorem to obtain the following result.

\begin{Theorem}\label{t:one_tree_labelled_recovery}
    Let $T$ be a labelled phylogenetic tree. Then $T$ has a unique $k$-multideck for $k \geq 3$.
\end{Theorem}

\begin{proof}
    Using a similar procedure to that of Theorem \ref{smaller_bucket_from_larger}, we can obtain the $3$-multideck of $T$ from its $k$-multideck. Then, by Theorem \ref{t:3-recoverable} there is a unique $T$ corresponding to this $3$-multideck. Hence there must have been a unique tree that produced the $k$-multideck.
\end{proof}

We additionally also observe that if one input tree is identified in the $k$-bucket, we can immediately identify the other as well.

\begin{Lemma}\label{labelled_buckets_one_from_other}
    Let $T_1$ and $T_2$ be labelled rooted phylogenetic trees on $n$ leaves. If $T_1$ is recoverable from $B_k(\{T_1, T_2\})$, then both $T_1$ and $T_2$ are recoverable from their $k$-bucket.
\end{Lemma}

\begin{proof}
    As $T_1$ is recoverable from $B_k(\{T_1, T_2\})$, we can completely determine the $k$-multideck of $T_1$ (by simply considering the $k$-leaf subtrees of $T_1$). As the $k$-bucket of $T_1$ and $T_2$ is formed from the union of the $k$-multidecks of the two trees, if we remove the $k$-multideck of $T_1$ from the $k$-bucket of $T_1$ and $T_2$, we obtain the $k$-multideck of $T_2$, which by Theorem \ref{t:one_tree_labelled_recovery} allows us to recover $T_2$. Thus $T_1$ and $T_2$ are recoverable from their $k$-bucket.
\end{proof}

We also recall the useful fact that any tree can be reconstructed from its $(n-1)$-multideck. The following Theorem is a rewording of Corollary 3.13 from \cite{Decks_of_rooted_binary_trees}.

\begin{Theorem}[\cite{Decks_of_rooted_binary_trees}]\label{reconstruct_from_n-1_multideck}
    Let $T$ be an unlabelled phylogenetic tree on $n \neq 4$ leaves. Then $T$ is the unique phylogenetic tree with its $(n-1)$-multideck.
\end{Theorem}

Note that in the case where $n < 4$, this is trivially true, as only a single unlabelled rooted tree with $n$ leaves exists in these cases. When $n=4$, there are 2 trees on $n$ leaves, but only a single tree on $n-1=3$ leaves, so the two trees on $4$ leaves have identical $(n-1)$-multidecks.
We will also frequently use the fact that if we can reconstruct one of the input trees for a $k$-bucket, finding the other input tree is a simple matter of collecting the remaining subtrees and then reconstructing the second tree from them.

\begin{Lemma}\label{unlabelled_buckets_one_from_other}
    Let $T_1$ and $T_2$ be unlabelled phylogenetic trees on $n \ne 4$ leaves. If $T_1$ is recoverable from $B_{n-1}(\{T_1, T_2\})$, then $T_1$ and $T_2$ are recoverable from their $(n-1)$-bucket.
\end{Lemma}

\begin{proof}
This proof proceeds almost identically to Theorem \ref{labelled_buckets_one_from_other}, except $T_2$ is recovered due to Theorem \ref{reconstruct_from_n-1_multideck}.
\end{proof}

We will use the standard notation to refer to subtrees on three leaves. All binary phylogenetic trees on three leaves have an isomorphic tree shape, in which one pendant subtree is a cherry, and the other pendant subtree is a single leaf. We will use the notation $ab|c$ to denote the $3$-leaf tree in which $(a,b)$ is the cherry and $c$ is the single leaf. This subtree will be referred to as a \emph{rooted triple}. The following are standard reconstruction definitions, adapted from \cite{bryant1995extension}.

\begin{Definition}
    Let $r = ab|c$ be a rooted triple. Then denote by $\supp(r)$ the set $\{a,b,c\}$. If $R=\{r_1, \dots, r_k\}$ is a set of rooted triples, define $\supp(R)=\cup_i^k \supp(r_i)$. More generally, let $T$ be a labelled phylogenetic tree, and denote by $\supp(T)$ the set of leaves of $T$, and if $S=\{T_1,\dots,T_k\}$ is a set of labelled phylogenetic trees define 
    \[\supp(S)=\cup_{i=1}^k \supp(T_i).\]
\end{Definition}

\begin{Definition}
    Let $R$ be a set of rooted triples and denote by $\langle R \rangle$ the set of phylogenetic trees on $\supp(R)$ that display all triples in $R$. The \emph{closure} of $R$ is $\text{cl}(R) = \cap_{T \in \langle R \rangle} M_3(T)$. We write $R \vdash ab|c$ if and only if $ab|c \in \text{cl}(R)$.
\end{Definition}

Bryant and Steel \cite{bryant1995extension} define a closure operation for sets of triples and quartets. We generalise this in the present context for sets of subtrees.

\begin{Definition}
    Let $S$ be a set of labelled phylogenetic trees on $k$ leaves and denote by $\langle S \rangle$ the set of phylogenetic trees on $\supp(S)$ that display all subtrees in $S$. The \emph{closure} of $S$ is $\text{cl}(S) = \cap_{T \in \langle S \rangle} M_k(T)$.
\end{Definition}

That is, if $S$ is a set of labelled phylogenetic trees on $k$ leaves, the closure of $S$ is the set of labelled phylogenetic trees on $k$ leaves that are displayed by every labelled phylogenetic tree that displays every tree in $S$. Note, of course, that if some subset of $S$ cannot all be displayed by the same tree, the closure of $S$ is the empty set.

\begin{Observation}
    \label{o:infer.closed}
    Consider a pair of trees $T_1$ and $T_2$, and their $k$-bucket. Suppose $B_k$ does not have a unique pre-image for this pair, that is, there is another pair of trees, $T_3$ and $T_4$, so that $B_k(\{T_1,T_2\})=B_k(\{T_3,T_4\})$.

    Consider $M_k(T_1) \cap M_k(T_3)$, which is a proper subset of $M_k(T_1)$ (i.e. non-empty and not equal to $M_k(T_1)$ itself), as otherwise we are forced to have $\{T_1,T_2\}=\{T_3,T_4\}$. Observe $M_k(T_2) \cap M_k(T_4)$ must have the same set of supports as $M_k(T_1) \cap M_k(T_3)$. Similarly, $M_k(T_2) \cap M_k(T_3)$ must have the same set of supports as $M_k(T_1) \cap M_k(T_4)$.

    So, to find such an example, we need to find a pair of trees with a non-trivial subset $S$ of their $k$-multidecks that we can fix, and then swap over all remaining elements of their multideck and obtain a new pair of trees. In particular, this subset must be inferentially closed, in the sense that $\text{cl}(S) =S$. For example, if $k=3$, $ab|c$ and $bc|d$ together imply $ac|d$, so our fixed subset could not be $\{ab|c,bc|d\}$ without $ac|d$.
\end{Observation}

While we do not explicitly use the above observation in the present manuscript, the authors found it philosophically useful in generation of proofs and counterexamples.

\subsection{Recoverability of statistics}

\subsubsection{Recoverability from sufficiently different statistics}

For the remainder of this section, we consider a number of tree statistics and some circumstances in which the corresponding functions are recoverable, including the height of the trees and a new statistic that we refer to as their root balance. All results in this section specifically consider only $(n-1)$-buckets, but can likely be extended. To avoid being overly verbose we will refer to recoverability of certain statistics, which should be interpreted as referring to the recoverability of the function that takes a tree as input and returns that statistic. 

We first observe that some properties can be recovered if the pair of trees that form the bucket are sufficiently structurally different. We note that these results are somewhat orthogonal to the remainder of the results in this manuscript and do not form part of any subsequent results, but we include them in case such statistics and proof methods may prove useful in generalisations of the main problem of this manuscript. 

\begin{Lemma}
    If the heights of two phylogenetic trees $T_1$ and $T_2$ (labelled or unlabelled) differ by at least $3$, then $T_1$ and $T_2$ are recoverable from their $(n-1)$-bucket.
\end{Lemma}

\begin{proof}
    We proceed by showing that if $|h(T_1) - h(T_2)| \geq 3$, the $k$-bucket $B_{n-1}(\{T_1, T_2\})$ can be unambiguously partitioned into the individual multidecks $M_{n-1}(T_1)$ and $M_{n-1}(T_2)$.
    
    When a single leaf is removed from a binary phylogenetic tree, the height of the resulting subtree either decreases by exactly $1$ or remains unchanged. Therefore, the possible heights of the subtrees in $M_{n-1}(T_1)$ belong to the set $\{h(T_1), h(T_1) - 1\}$, and the possible heights of the subtrees in $M_{n-1}(T_2)$ belong to $\{h(T_2), h(T_2) - 1\}$.

    Thus, we know that any given subtree of height $h$ must have originated from a tree of height $h$ or $h + 1$. Because $|h(T_1) - h(T_2)| \geq 3$, the sets of possible subtree heights for $T_1$ and $T_2$ are separated by a gap of at least 2. Consequently, we can definitively group the subtrees in $B_{n-1}(\{T_1, T_2\})$ into two distinct clusters: subtrees originating from the same tree will have heights that differ by at most 1 (that is, they are equal or consecutive integers), whereas subtrees from different trees will have heights differing by at least 2. 
    
    This separation allows us to unambiguously assign every subtree to either $M_{n-1}(T_1)$ or $M_{n-1}(T_2)$, from which the original trees $T_1$ and $T_2$ can then be independently reconstructed by either Theorem \ref{t:one_tree_labelled_recovery} or Theorem \ref{reconstruct_from_n-1_multideck} (depending on whether $T_1$ and $T_2$ are labelled or unlabelled).
\end{proof}

\begin{Remark}
    To show that the above condition is strict  we refer to Figure \ref{fig:six_leaf_unlabelled_trees_counterexamples} which illustrates that trees $6a$ and $6b$ have heights differing by 2, but have the same 5 leaf subtrees as 6c and 6d. 
\end{Remark}

Similar proofs can easily be constructed for any statistic that is sufficiently stable under the deletion of a single leaf, for instance, the number of cherries.

\subsubsection{Fully recoverable}

We now show that some further statistics are completely recoverable - that is, for any pair of input trees we may recover the (multi)set of results of applying the relevant function to these trees. In particular, we show a stronger result - that for the statistics considered in this section, one can recover the multiset of that tree statistic for the input trees, by considering only the multiset of that statistic for the trees in the bucket (as can be seen by considering the fact that the proofs only use this information). For instance, for the statistic \emph{root balance}, we can recover the root balances of the input trees given only the root balances of the trees in the $(n-1)$-bucket. We note this fact in particular as later we will show recoverability for almost all pairs of trees from their $(n-1)$-bucket, so of course any statistic can be recovered given all information from their $(n-1)$-bucket.

We start with considering root balance, which is a statistic that we introduce and forms the basis of the majority of proofs in Section \ref{s:labeln-1}.

\begin{Definition}
    Let $T$ be a rooted binary phylogenetic tree with at least 2 leaves. The \emph{major and minor subtrees} of $T$ are the two maximal subtrees of $T$ rooted at the children of the root of $T$, such that the major subtree has at least as many leaves as the minor subtree. The \emph{root balance} of $T$ is the number of leaves in $T$'s minor subtree.
\end{Definition}

We first note that the maximum value for the root balance of a tree with $n$ leaves is of course $\lfloor \frac{n}{2} \rfloor$.

We additionally note that several similar concepts appear in the literature. For instance, \cite{Decks_of_rooted_binary_trees} defines the \emph{root-split} of a tree $T$ to be the set $S=\{|T_1|,|T_2|\}$, where $T_1$ and $T_2$ are the two subtrees rooted at the children of the root of $T$, and $| \cdot |$ indicates the number of leaves in a tree. In this case, the root balance is just $\min(S)$. Additionally, in the \texttt{treestats} R package \cite{janzen2024phylogenetic}, root imbalance is defined to be $\frac{|T_1|}{|T_1|+|T_2|}$, where, to use the language of the present manuscript, $T_1$ is the major subtree and $T_2$ the minor subtree. In this case, denoting the  root balance of a tree $T$ on $n$ leaves by $r$, the  root imbalance of $T$ would simply be $1-\frac{r}{n}$.

Root balance is especially helpful for recovering the input trees of a bucket, due to the fact that it is fully recoverable and $(n-1)$-leaf subtrees of an $n$-leaf tree have root balances that differ from the full tree in very predictable ways. 

We now introduce a lemma for counting the number of subtrees with certain root balances, which will subsequently be used to show full recoverability of root balance for $(n-1)$-buckets.

\begin{Lemma}[Counting Lemma]\label{unlabelled_buckets_counting}
    Let $T$ be a binary rooted phylogenetic tree with $n$ leaves and root balance $k$. Then $M_{n-1}(T)$ contains:
    \begin{itemize}
        \item $n-k$ subtrees with root balance $k$ and $k$ subtrees with root balance $k-1$ if $1<k<\frac{n}{2}$.
        \item $n$ subtrees with root balance $k-1$ if $k=\frac{n}{2}$.
        \item $n-1$ subtrees with root balance $1$, and one subtree with root balance at most $\lfloor\frac{n-1}{2}\rfloor$ if $k=1$.
    \end{itemize}
\end{Lemma}

\begin{proof}
    Suppose $1<k<\frac{n}{2}$. Then $T$ has a minor subtree with $k$ leaves, and a major subtree with $n-k$ leaves. Each subtree of $T$ with $n-1$ leaves is obtained by removing some leaf of $T$. If this leaf is in the major subtree of $T$, then, as the number of leaves of the minor subtree is unchanged and removing one leaf from the major subtree will not make it have fewer leaves than the minor subtree, the resulting subtree will have root balance $k$. If this leaf is in the minor subtree of $T$, as $k \geq 2$ there is still at least one leaf in the minor subtree of the resulting subtree. Thus $M_{n-1}(T)$ contains $n-k$ subtrees with root balance $k$ and $k$ subtrees with root balance $k-1$ if $1<k<\frac{n}{2}$.
    
    Now suppose $k = \frac{n}{2}$. Then the minor and major subtrees of $T$ both have $k$ leaves. Each subtree of $T$ with $n-1$ leaves is obtained by removing some leaf of $T$. No matter which of the $n$ leaves are chosen, this will result in a subtree with a major subtree with $k$ leaves, and a minor subtree with $k-1$, where the major subtree is that major/minor subtree of $T$ from which the leaf was not removed. Thus $M_{n-1}(T)$ contains $n$ subtrees with root balance $k-1$ if $k=\frac{n}{2}$.
    
    Now suppose $k=1$. Then $T$ has a major subtree with $n-1$ leaves, and a minor subtree consisting of a single leaf. Each subtree of $T$ with $n-1$ leaves is obtained by removing some leaf of $T$. If this leaf is in the major subtree of $T$, then the resulting subtree will still have a minor subtree with only a single leaf, and so will have root balance 1. Otherwise the leaf removed must be the single leaf in the minor subtree, resulting in a subtree that is the major subtree of $T$, which has no restrictions on its shape, and thus could have any root balance possible for a tree with $n-1$ leaves, which is at most $\lfloor\frac{n-1}{2}\rfloor$. Thus $M_{n-1}(T)$ contains $n-1$ subtrees with root balance $1$, and one subtree with root balance at most $\lfloor\frac{n-1}{2}\rfloor$ if $k=1$.
    This covers all possible cases for the value of $k$.
\end{proof}

We can now show that root balance is always recoverable from $(n-1)$-buckets for two trees.

\begin{Theorem}\label{unlabelled_buckets_rb_equal}
    Let $T_1$ and $T_2$ be binary rooted phylogenetic trees, on $n>4$ leaves with root balances $k_1$, and $k_2$ respectively. Then the set $\{k_1, k_2\}$ is recoverable from their $(n-1)$-bucket.
\end{Theorem}

\begin{proof}
    Consider the multiset, $R$, of the root balances of trees in $B_{n-1}(\{T_1, T_2\})$. Consider the multiplicity of the maximum value, $k$, in $R$. Suppose $k = \frac{n}{2}-1$. By Lemma \ref{unlabelled_buckets_counting} an $(n-1)$-leaf subtree with root balance $\frac{n}{2}-1$ must be a subtree of a tree with root balance $1$, $\frac{n}{2}-1$ or $\frac{n}{2}$. A tree with root balance $1$ will have at most $1$ subtree with root balance $\frac{n}{2}-1$, a tree with root balance $\frac{n}{2}-1$ will have exactly $\frac{n}{2}+1$, and a tree with root balance $\frac{n}{2}$ will have $n$. Thus if there are $1$ or $2$ subtrees with root balance $k$ then one of $k_1$ or $k_2$ is $1$. If there are $\frac{n}{2}+1$, $n+2$ or $\frac{3n}{2}+1$ subtrees with root balance $k$ than one of $k_1$ or $k_2$ is $\frac{n}{2}-1$. And finally if there are $n$ or $2n$ subtrees with root balance $k$ then one of $k_1$ or $k_2$ is $\frac{n}{2}$.
    
    By Lemma \ref{unlabelled_buckets_counting} a subtree with $n-1$ leaves and root balance $k \neq \frac{n}{2}-1$ must either be a subtree of a tree with root balance $1$ or $k$. A tree with $n>4$ leaves and root balance $k$ must have at least $3$ subtrees with root balance $k$, whereas with root balance $1$ it can have at most 1, by Lemma \ref{unlabelled_buckets_counting}. Thus by considering the cardinality of $k$ in $R$ we can determine $k_1$ or $k_2$, say $k_1$ without loss of generality. Additionally this determines the total contribution to $R$ by subtrees of $T_1$, allowing us to determine the exact multiset of the root balances of subtrees of $T_2$. By Lemma \ref{unlabelled_buckets_counting}, if there are at least $3$ subtrees with root balance $1$ then $k_2 = 1$. Otherwise if all subtrees have the same root balance then $k_2 = \frac{n}{2}$. And finally if neither of these are true then $k_2$ is the maximum root balance of any of these trees. Thus we can determine both $k_1$ and $k_2$, and so the set $\{k_1, k_2\}$ is recoverable.
\end{proof}

We now introduce a novel tree statistic that will be used in the proof of Theorem \ref{unlabelled_buckets_theorem5}, and show that it is also fully recoverable from an $(n-1)$-bucket.

\begin{Definition}
    Let $T$ be a rooted binary phylogenetic tree. If the root of $T$ has no leaf child, define the \emph{pendant depth of $T$} to be $0$. Suppose $T$ has a maximal path from the root to some vertex consisting of $k$ vertices that each have a leaf child. Then we refer to the path as the \emph{pendant path} and say $T$ has a \emph{pendant depth} of $k$. Note that in the case that the path has only $1$ vertex, that is, it consists only of the root, this means that the root has a leaf child and no child of the root has a leaf child.
\end{Definition}

\begin{Observation}
    A tree with pendant depth $k$ has a path from the root consisting of $k$ internal vertices with a leaf child. The $(k+1)$-th vertex, if it exists and is not a single vertex, is then the root of a subtree with root balance greater than $1$. Refer to this subtree as the \emph{highest non-pendant subtree}.
\end{Observation}

We begin with a counting lemma in a similar style to the one derived for root balance.

\begin{Lemma}[Counting Lemma]\label{pendant_depth_lemma}
    Let $T$ be a rooted binary phylogenetic tree with root balance $1$. Then $T$ has a non-zero pendant depth $k$. Furthermore, either 
    
    \begin{enumerate}
        \item $T$ is a caterpillar tree and the $(n-1)$-multideck consists of $n$ subtrees of pendant depth $n-2$, or
        \item $T$ is not a caterpillar and the $(n-1)$-multideck of $T$ consists of $k$ subtrees of pendant depth $k-1$, and $n-k$ subtrees of pendant depth at least $k$, of which zero, two or four of which are strictly greater than $k$. Furthermore, if there are two or four with pendant depth strictly greater than $k$, they each have the same pendant depth.
    \end{enumerate} 
\end{Lemma}

\begin{proof}
    As $T$ has root balance $1$, the root has a leaf child. Hence the pendant depth of $T$ is at least $1$.

    First suppose $T$ is a caterpillar tree, and hence has a pendant depth of $n-1$. Every subtree in the $(n-1)$-multideck of a caterpillar tree is the unique $(n-1)$-leaf caterpillar tree, each with a pendant depth of $n-2$.

    Now suppose $T$ is not a caterpillar tree, and thus the highest non-pendant subtree exists. Now, consider the set of subtrees obtained by suppressing a vertex on the pendant path and deleting its leaf child. Then each such tree will still have a pendant path consisting of $k-1$ vertices and hence has a pendant depth of $k-1$. As there are $k$ options for internal vertices that are deleted, there will be $k$ such subtrees.

    Now consider the subtrees obtained by suppressing some parent of a leaf vertex in the highest non-pendant subtree and deleting its leaf child. In this case, the $k$ internal vertices on the pendant path still all exist and so the resulting subtree has pendant length at least $k$. As there are $n-k$ options for the leaf deletion, there will be $n-k$ such subtrees.

    Now, in order for the pendant length to increase, the non-leaf child $v$ of the $k$-th vertex on the pendant path must have one pendant subtree consisting only of a cherry (say, $c$ and its two leaf children), so that upon suppression of $c$ and deletion of one of its leaf children $v$ has a leaf child. In particular, such a cherry will produce two subtrees in the $(n-1)$-multideck of pendant length greater than $k$, as there are two possible leaf children of $c$ to be deleted. As $v$ can have at zero, one or two pendant subtrees consisting of cherries, there are therefore zero, two or four possible subtrees with pendant length greater than $k$ (and, in fact, only exactly four precisely when the highest non-pendant subtree is the fully balanced tree on four leaves). In each case for two or four subtrees with strictly greater than $k$ pendant depth, the resulting trees will be isomorphic in tree shape, so will have the same pendant depth.

    Hence the lemma is proven.
\end{proof}

We can now provide sufficient conditions for pendant depth to be recoverable. We note that the requirement that $n>13$ can likely be relaxed, but it seems that reducing this number would require significantly more arguments and likely numerous additional cases.

\begin{Theorem}\label{t:pendant.depth.recoverable}
    Let $T_1$ and $T_2$ be rooted phylogenetic trees, on $n>13$ leaves with pendant depths $k_1$ and $k_2$ respectively. Then the set $\{k_1, k_2\}$ is recoverable from their $(n-1)$-bucket.
\end{Theorem}

\begin{proof}
    We observe that the only trees with a pendant depth of zero that have trees with non-zero pendant depth in their $(n-1)$-multideck are those trees with root balance $2$, as any tree has non-zero pendant depth if and only if it has root balance $1$. In the case of trees with root balance $2$, such a tree will have precisely two trees of non-zero pendant-depth in their $(n-1)$-multideck by Lemma \ref{pendant_depth_lemma}.

    We first consider the case in which the $(n-1)$-bucket contains at least one tree of pendant depth zero. Combining the above observation and Lemma \ref{pendant_depth_lemma}, there are three possibilities for a multideck of a single tree that contains at least one subtree of pendant depth zero - it is a tree with root balance $3$ or greater and all $n$ trees in its $(n-1)$-multideck have pendant depth zero; it is a tree with root balance precisely $2$ and its $(n-1)$-multideck has $n-2$ of pendant depth $0$ and $2$ of pendant depth at least $1$, or it has root balance and pendant depth exactly $1$, so it has one tree of pendant depth $0$ and $(n-1)$ subtrees of pendant depth at least $1$.

    It follows that an $(n-1)$-bucket containing at least one tree with zero pendant depth has one of the following nine options for number of trees of zero pendant depth, denoting root balances by $r_1$ and $r_2$ and supposing $r_1 \le r_2$ without loss of generality: all $2n$ ($r_1,r_2\ge 3$), $2n-2$ ($r_1=2,r_2=3$), $2n-4$ ($r_1=r_2=2$), $n+1$ ($r_1=1,r_2=3$), $n$ ($r_2=3$ and the other tree has pendant depth greater than $1$), $n-1$ ($r_1=1,r_2=2$), $n-2$ ($r_2=2$ and the other tree has pendant depth greater than $1$), $2$ ($r_1=r_2=1$ and both have pendant depth $1$), $1$ ($r_1r_2=1$, one tree has pendant depth $1$, and the other has pendant depth greater than $1$). These are all discrete values when $n \ge 6$, and hence pendant depth is always recoverable for $k_1$ or $k_2 \in \{0,1\}$.

    We can therefore assume that both inputs have pendant depth at least $2$.

    We now consider caterpillar subtrees. Observe that if an input is not a caterpillar, it can have at most four subtrees that are caterpillars by Lemma \ref{pendant_depth_lemma}. Hence, we first check if the highest pendant depth of a subtree is $n-2$. If there are $0$, there are no caterpillars in the input trees and if there are $2n$, both input trees were caterpillars (with pendant depth $(n-1)$). If there are at least $n$, we know one of the input trees are caterpillars (since $n>8$ and two non-caterpillars can contribute at most $8$ caterpillar subtrees to the $(n-1)$-bucket). We can then recover the other pendant depth by considering the smallest pendant depth in the bucket, which must be at least $1$ since we assume input trees have pendant depth at least two, and this smallest pendant depth must be $k-1$ for a pendant depth of $k$ by Lemma \ref{pendant_depth_lemma}. Hence we recover $\{n-1,k\}$.
    
    We can now assume that both input trees have pendant depth at least $2$ and that neither input tree is a caterpillar.

    We now consider the smallest pendant depth of a subtree in the $(n-1)$-bucket. As both trees have pendant depth at least $2$, there is some tree with a minimal pendant depth of $k-1>0$, from which we can immediately identify that some input tree had a pendant depth of $k$. If there are precisely $k$ of them, we know there are two distinct pendant depths for the input trees - if there are more than $k$, we know that the input trees had pendant depths $\{k,k\}$. 

    We then check if there are at least $j$ subtrees with pendant depth $j-1$ for each $j \ge 6$ up to the maximum pendant depth of the subtrees in the $(n-1)$-bucket. If there are, by Lemma \ref{pendant_depth_lemma} $j$ of these can only have come from an input tree with pendant depth of $j$, and so we have recovered $\{j,k\}$ for all cases where at least one is six or greater.

    Hence the only remaining cases are those where the input trees have pendant depth $2 \le j,k \le 5$ and $j \ne k$, and without loss of generality assume $j>k$ (and we always know the value of the smaller pendant depth).

    If $j=k+1$, we simply consider the number of trees with pendant depth $k$. The tree of pendant depth $k$ will contribute $(n-k), (n-k-2)$, or $(n-k-4)$ subtrees of pendant depth $k$, while the tree of pendant depth $k+1$ will contribute $k+1$, for a total of $(n+1), (n-1)$ or $(n-3)$ subtrees of pendant depth $k$. In any case, if $j>k+1$ it would not contribute any subtrees with pendant depth $k$, so this is the only case where the number of subtrees of pendant depth $k$ have a different parity to $n$, and so this case is recoverable.

    Finally, we know that the input tree of pendant depth $j$ must contribute at least $n-j-4>n-9$ (since $j \le 5$) subtrees of pendant depth $j$. Since $n>13$, there will be more than four subtrees, so not all can have been produced by the tree with pendant depth $k<j$ and its none of the pendant depth $k$ subtrees produced by the $k$ input tree since $j>k+1$. Hence we can then immediately recover $j$ too, as it is simply the pendant depth of those (at least) $n-9$ subtrees.
\end{proof}

\section{Unlabelled \texorpdfstring{$(n-1)$}{(n-1)}-buckets}\label{s:unlabelledn-1}

We will now systematically characterise the recoverability of unlabelled trees from their $(n-1)$-buckets. 

The majority of this section will establish recoverability for pairs of trees with at least a certain number of leaves, dependent on their root balance. In particular, the main theorem of this section is the following theorem.

\begin{Theorem}
    \label{always_recoverable_unlabelled}
    Let $T_1$ and $T_2$ be unlabelled  rooted binary phylogenetic trees on $n>8$ leaves, $n<4$ leaves, or $n=5$ leaves. Then $T_1$ and $T_2$ are recoverable from their $(n-1)$-bucket.
\end{Theorem}

The proof of this theorem will be achieved by considering a number of subcases. We summarise the results of the subcases in Table \ref{tab:unlabelled.cases}.

First, we require two technical lemmas.

\begin{Lemma}\label{recover_from_root_balance_class}
    Let $T$ be an unlabelled phylogenetic tree with major and minor subtrees $T_A$ and $T_B$ (not necessarily respectively). Suppose that a submultiset $S$ of $M_{n-1}(T)$ is the multiset of subtrees obtained from $T$ by deleting a leaf from $T_A$. If $T_A$ does not have $4$ leaves then $T$ is uniquely determined by $S$.
\end{Lemma}

\begin{proof}
    As each subtree is $S$ is obtained by removing a leaf from $T_A$ in $T$, each must have $T_B$ as a major or minor subtree. Thus $T_B$ is uniquely determined by $S$. Additionally if we consider the minor or major subtrees of subtrees in $S$ which are not $T_B$, they form the $(k-1)$-multiset of $T_A$, where $T_A$ has $k$ leaves. But as $T_A$ does not have $4$ leaves, by Theorem \ref{reconstruct_from_n-1_multideck}, $T_A$ is uniquely determined by $S$. Finally as both $T_A$ and $T_B$ are uniquely determined by $S$, $T$ is uniquely determined by $S$ as the unique tree with $T_A$ and $T_B$ as major and minor subtrees.
\end{proof}
\begin{Lemma}\label{recover_from_pendant_subtrees}
Let $T_1$ and $T_2$ be unlabelled phylogenetic trees with major and minor subtrees $(A_1,B_1)$ and $(A_2,B_2)$. Suppose the $(n-1)$-bucket determines the multisets
\[
\{A_1,A_2,B_1,B_2\}.
\]
If $A_1 \notin \{A_2, B_1, B_2\}$, none of $A_2$, $B_1$, or $B_2$ have a single leaf more than $A_1$, and $A_2$ does not have $4$ leaves, then the pairing between major and minor subtrees is uniquely determined by their multidecks.
Consequently $T_1$ and $T_2$ are recoverable from their $(n-1)$-bucket.
\end{Lemma}
\begin{proof}
    Let $S$ denote the multiset of subtrees in $B_{n-1}(\{T_1, T_2\})$ with $A_1$ as a major or minor subtree. Then each subtree in $S$ must be obtained by removing a leaf from $B_1$, as there is no way to remove a leaf from any of the major or minor subtrees of $T_1$ and $T_2$ to obtain $A_1$, thus $A_1$ must be the major or minor subtree of the tree a leaf was removed from, but not the subtree the leaf was removed from. Then the multiset of major or minor subtrees of the subtrees in $S$ that are not $A_1$ form the $(k-1)$-multiset of $A_2$ where $A_2$ has $k$ leaves. Then, as $A_2$ does not have $4$ leaves by Theorem \ref{reconstruct_from_n-1_multideck}, $T_1$ is recoverable from $B_{n-1}(\{T_1, T_2\})$. Thus by Theorem \ref{unlabelled_buckets_one_from_other}, $T_1$ and $T_2$ are recoverable.
\end{proof}

The following two theorems establish recoverability of any pair of trees with at least $4$ leaves from their $k$-buckets, with certain restrictions on their root balances.

\begin{table}
    \centering
    \begin{tabular}{c c c}
\textbf{Root balances} & \textbf{Number of Leaves for Recoverability} & \textbf{Theorem}\\
\hline
$1 < k_1 < k_2$ & $n>4$ & \ref{unlabelled_buckets_theorem1} \\
$1 < k_1 = k_2 < \lfloor \frac{n}{2} \rfloor$ & $n>4$ & \ref{unlabelled_buckets_theorem2} \\
$k_1=k_2=\frac{n}{2} $ & $n>8, n$ even & \ref{unlabelled_buckets_theorem3}\\
$k_1 = k_2 = \lfloor \frac{n}{2} \rfloor$ & $n>7, n$ odd & \ref{unlabelled_buckets_theorem3.5}\\
$k_1=1,k_2>2$ & $n>5$ & \ref{unlabelled_buckets_theorem4}\\
$k_1=k_2=1$ & $n>4$ & \ref{unlabelled_buckets_theorem5}\\
$k_1=1,k_2=2$ & $n>7$ & \ref{unlabelled_buckets_theorem6}\\
\end{tabular}
\caption{The cases considered in this section}
\label{tab:unlabelled.cases}
\end{table}

\begin{Theorem}\label{unlabelled_buckets_theorem1}
    Let $T_1$ and $T_2$ be unlabelled phylogenetic trees on $n>4$ leaves with root balance $k_1$ and $k_2$ respectively. If $1 < k_1 < k_2$ then $T_1$ and $T_2$ are recoverable from their $(n-1)$-bucket.
\end{Theorem}

\begin{proof}
    By Theorem \ref{unlabelled_buckets_rb_equal}, the root balances of $T_1$ and $T_2$ are recoverable. 
    
    Let the set of subtrees in $B_{n-1}(\{T_1, T_2\})$ with root balance $k_1-1$ be denoted by $S$. By Lemma \ref{unlabelled_buckets_counting}, $T_2$ has no subtrees with root balance $k_1-1$. Thus the subtrees in $S$ must be obtained by removing a leaf from the minor subtree of $T_1$, so we may apply Lemma \ref{recover_from_root_balance_class} and see that if $k_1 \neq 4$, then $T_1$ is recoverable from $B_{n-1}(\{T_1, T_2\})$, and as one input tree has been determined, by Lemma \ref{unlabelled_buckets_one_from_other}, $T_1$ and $T_2$ are recoverable from their $(n-1)$-bucket.
    
    Now we consider the case where $k_1=4$. In this case $k_2$ is at least $5$. Let the set of subtrees in $B_{n-1}(\{T_1, T_2\})$ with root balance $k_2$ be denoted by $S'$. Thus the subtrees in $S'$ must be obtained by removing a leaf from the major subtree of $T_2$. As $k_2 > 4$, by Lemma \ref{recover_from_root_balance_class} $T_2$ is recoverable from $B_{n-1}(\{T_1, T_2\})$. Then by Lemma \ref{unlabelled_buckets_one_from_other}, $T_1$ and $T_2$ are recoverable from their $(n-1)$-bucket.
    
\end{proof}

\begin{Theorem}\label{unlabelled_buckets_theorem2}
    Let $T_1$ and $T_2$ be unlabelled phylogenetic trees on $n>4$ leaves with root balance $k_1$ and $k_2$ respectively. If $1 < k_1 = k_2 < \lfloor\frac{n}{2}\rfloor$ then $T_1$ and $T_2$ are recoverable from their $(n-1)$-bucket.
\end{Theorem}

\begin{proof}
    By Theorem \ref{unlabelled_buckets_rb_equal}, the root balances of $T_1$ and $T_2$ are recoverable.

    Let $T_A$ and $T_B$ be the major and minor subtrees of $T_1$ respectively. Let $T_C$ and $T_D$ be the major and minor subtrees of $T_2$ respectively.
    
    Let the subtrees in $B_{n-1}(\{T_1, T_2\})$ with root balance $k_1-1$ be denoted by $S$. By Lemma \ref{unlabelled_buckets_counting}, exactly half of the subtrees in $S$ must be obtained by removing a leaf from the minor subtree of $T_1$, and the other half from the minor subtree of $T_2$. Additionally, each subtree in $S$ must have either $T_A$ or $T_C$ as a major subtree, depening on which of $T_1$ or $T_2$ they are a subtree of. Thus half the major subtrees in $S$ have $T_A$ as a major subtree, and the other half have some tree $T_C$ as a major subtree. Thus $T_A$ and $T_C$ are recoverable from $B_{n-1}(\{T_1, T_2\})$.

    Let the subtrees in $B_{n-1}(\{T_1, T_2\})$ with root balance $k_1$ be denoted by $S'$. By a similar argument, half the minor subtrees in $S'$ have some tree $T_B$ as a minor subtree, and the other half have some tree $T_D$ as a minor subtree. Thus $T_B$ and $T_D$ are recoverable from $B_{n-1}(\{T_1, T_2\})$.

    If $T_A = T_C$ then $T_1$ and $T_2$ are the unique trees with $T_A$ as the major subtree and $T_B$ and $T_D$ as minor subtrees, and so $T_1$ and $T_2$ are recoverable from their $(n-1)$-bucket. Similarly if $T_B = T_D$ then $T_1$ and $T_2$ are recoverable from their $(n-1)$-bucket.

    Now suppose $T_A \neq T_C$ and $T_B \neq T_D$. For $T_B$ and $T_D$ to differ they must have at least $4$ leaves. Thus $T_A$ has greater than $4$ leaves. As $T_A \neq T_C$, we must have that $T_A \notin \{T_B, T_C, T_D\}$. Additionally none of $T_B$, $T_C$ or $T_D$ have more leaves than $T_A$. Thus by Lemma \ref{recover_from_pendant_subtrees}, $T_1$ and $T_2$ are recoverable from their $(n-1)$-bucket.
\end{proof}

The following theorem concerns the case where the major and minor subtrees of each input tree have an equal number of leaves. Unlike our previous theorems which assume our trees have more than 4 leaves, this theorem assumes our trees have more than 8 leaves. Later, in Theorem \ref{unlabelled_buckets_large_counterexamples} we will see that there exist unrecoverable pairs of balanced trees on 8 leaves. This is due to the fact that the left and right subtrees of a balanced tree with 8 leaves will each have 4 leaves, meaning that they cannot be recovered from their $(n-1)$-multideck due to Theorem \ref{reconstruct_from_n-1_multideck}.

\begin{Theorem}\label{unlabelled_buckets_theorem3}
    Let $T_1$ and $T_2$ be unlabelled phylogenetic trees on $n>8$ leaves with root balance $k_1$ and $k_2$ respectively. If $n$ is even and $k_1 = k_2 = \frac{n}{2}$ then $T_1$ and $T_2$ are recoverable from their $(n-1)$-bucket.
\end{Theorem}

\begin{proof}
    By Theorem \ref{unlabelled_buckets_rb_equal}, the root balances of $T_1$ and $T_2$ are recoverable. 
    
    All subtrees in $B_{n-1}(\{T_1, T_2\})$ have root balance $k_1-1$ by Lemma \ref{unlabelled_buckets_counting}. Each of the major and minor subtrees of $T_1$ and $T_2$ are major subtrees of exactly a quarter of the subtrees in $B_{n-1}(\{T_1, T_2\})$. Thus by considering the major subtrees of subtrees in $B_{n-1}(\{T_1, T_2\})$, we can determine the major and minor subtrees of $T_1$ and $T_2$, say $T_A$, $T_B$, $T_C$, and $T_D$.
    
    If one of those subtrees is distinct from the others, say $T_A \notin \{T_B, T_C, T_D\}$, let $S$ be the set of subtrees in $B_{n-1}(\{T_1, T_2\})$ with major subtree $T_A$. We can then immediately observe that as $n>8$ and $n$ is even, all four of the major and minor subtrees have more than $4$ leaves. Furthermore, all four have the same number of leaves, and so none have one more leaf than $T_A$. Thus, we can apply Lemma \ref{recover_from_pendant_subtrees} and $T_1$ and $T_2$ are recoverable.
    
    Otherwise, if $T_A = T_B = T_C = T_D$ then $T_1$ and $T_2$ are both the unique tree with $T_A$ as both major and minor subtrees, and so $T_1$ and $T_2$ are recoverable.
    
    Finally, if $T_A = T_B \neq T_C = T_D$ either $T_1=T_2$, and have $T_A$ and $T_C$ as major and minor subtrees, or $T_1$ has 2 copies of $T_A$ as its major and minor subtrees, and similar for $T_2$ with $T_C$ or vice versa. Let $S$ denote the set of subtrees in $B_{n-1}(\{T_1, T_2\})$ with $T_A$ as their major subtree. Then if we consider the minor subtrees of trees in $S$, they form either the $(k_1-1)$-bucket of $T_A$ and $T_A$ or $T_C$ and $T_C$. But that is the $(k_1-1)$-multideck of $T_A$ or $T_C$ with multiplicities doubled, so by Lemma \ref{reconstruct_from_n-1_multideck} this uniquely determines whether $T_1$ is the unique tree with major and minor subtree $T_A$ and $T_2$ is the unique tree with major and minor subtrees $T_C$, or $T_1$ and $T_2$ are both the unique trees with $T_A$ and $T_B$ as major and minor subtrees. Thus $T_1$ and $T_2$ are recoverable from their $(n-1)$-bucket.
\end{proof}

\begin{Theorem}\label{unlabelled_buckets_theorem3.5}
    Let $T_1$ and $T_2$ be unlabelled phylogenetic trees on $n>7$ leaves with root balance $k_1$ and $k_2$ respectively. If $n$ is odd and $k_1 = k_2 = \lfloor\frac{n}{2}\rfloor$ then $T_1$ and $T_2$ are recoverable from their $(n-1)$-bucket.
\end{Theorem}

\begin{proof}
    By Theorem \ref{unlabelled_buckets_rb_equal}, the root balances of $T_1$ and $T_2$ are recoverable.
    
    Let $T_A$ and $T_C$ be the major subtrees of $T_1$ and $T_2$ respectively. Let $S$ be the multiset of subtrees in $B_{n-1}(\{T_1, T_2\})$ with root balance $k_1 - 1$. These must have been obtained from $T_1$ or $T_2$ by removing a leaf from their minor subtree, thus each has either $T_A$ or $T_C$ as their major subtree. As half will have $T_A$ and half $T_C$, these major subtrees are uniquely determined by $S$, and thus uniquely determined by $B_{n-1}(\{T_1, T_2\})$.
    
    Let $T_B$ and $T_D$ be the minor subtrees of $T_1$ and $T_2$ respectively. Let the subtrees in $B_{n-1}(\{T_1, T_2\})$ with root balance $k_1$ be denoted $S'$. These subtrees are balanced, with one major/minor subtree being $T_B$ or $T_D$, and one major/minor subtree being $T_A$ or $T_C$ with a leaf removed. Consider the multiset $M$ of the major/minor subtrees of trees in $S'$. One quarter of the trees in $M$ are $T_B$, a quarter are $T_D$, a quarter form the $k_1$-multiset of $T_A$, and the final quarter form the $k_1$-multiset of $T_C$. Thus the multiset $M - M_{k_1}(T_A)-M_{k_1}(T_C)$ contains exactly two trees, $T_B$ and $T_D$. As $T_A$ and $T_C$ are uniquely determined by $B_{n-1}(\{T_1, T_2\})$, so too are $T_B$ and $T_D$.
    
    If $T_A = T_C$, then $T_1$ and $T_2$ are the two unique trees with $T_A$ as a major subtree, and $T_B$ and $T_D$ as minor subtrees, and so $T_1$ and $T_2$ are recoverable. If $T_A \neq T_C$, then $T_A \notin \{T_B, T_C, T_D\}$, none of $T_B$, $T_C$ and $T_D$ have one more leaf than $T_A$, and as $n>7$, $T_A$ has more than $4$ leaves. Thus by Lemma \ref{recover_from_pendant_subtrees}, $T_1$ and $T_2$ are recoverable.
    
\end{proof}

\begin{Theorem}\label{unlabelled_buckets_theorem4}
    Let $T_1$ and $T_2$ be unlabelled phylogenetic trees on $n>5$ leaves with root balance $k_1$ and $k_2$ respectively. If $k_1 = 1$ and $k_2 > 2$ then $T_1$ and $T_2$ are recoverable from their $(n-1)$-bucket.
\end{Theorem}

\begin{proof}
        We can immediately recover the root balance by Theorem \ref{unlabelled_buckets_rb_equal}, so we know $k_1=1$ and $k_2>2$.

        Let the subtrees in $B_{n-1}(\{T_1, T_2\})$ with root balance $1$ be denoted $S$. By Lemma \ref{unlabelled_buckets_counting}, either $n$ or $n-1$ subtrees in $S$ must be subtrees of $T_1$, and none can be subtrees of $T_2$. Thus $|S| = n$ or $|S| = n-1$.

        If $|S| = n$, then $S$ is the $(n-1)$-multideck of $T_1$, as it contains all $n$ subtrees of $T_1$ with $n-1$ leaves. As $n>4$, by Theorem \ref{reconstruct_from_n-1_multideck}, this multideck uniquely determines $T_1$, so $T_1$ is recoverable. Then by Lemma \ref{unlabelled_buckets_one_from_other}, $T_1$ and $T_2$ are recoverable from their $(n-1)$-bucket.

        Now suppose $|S| = n - 1$. Then each subtree in $S$ must have been obtained from $T_1$ by removing one of the $n-1$ leaves in its major subtree. Thus if we consider the multiset of major subtrees of subtrees in $S$, they form the $(n-2)$-multideck of the major subtree of $T_1$. But $n>5$, so $n-1 > 4$. Thus by Theorem \ref{reconstruct_from_n-1_multideck} this multideck uniquely determines the major subtree of $T_1$. But as the minor subtree of $T_1$ must be a single leaf, this also uniquely determines $T_1$, so $T_1$ is recoverable. Then by Lemma \ref{unlabelled_buckets_one_from_other}, $T_1$ and $T_2$ are recoverable from their $(n-1)$-bucket.
\end{proof}

The next case is made surprisingly simple by considering the pendant depth statistics introduced in Section \ref{s:general}.

\begin{Theorem}\label{unlabelled_buckets_theorem5}
    Let $T_1$ and $T_2$ be unlabelled phylogenetic trees on $n>4$ leaves with root balance $k_1$ and $k_2$ respectively. If $k_1 = k_2 = 1$ then $T_1$ and $T_2$ are recoverable from their $(n-1)$-bucket.
\end{Theorem}

\begin{proof}
    We can immediately recover the root balance by Theorem \ref{unlabelled_buckets_rb_equal}. Therefore we immediately know that any input trees must have $k_1=k_2=1$.

    As $k_1=k_2=1$, both input trees must have a non-zero pendant depth, and hence must have multiplicity in the $(n-1)$-multideck as described by Lemma \ref{pendant_depth_lemma}.

    Consider some subtree in the $(n-1)$-multideck with minimum pendant depth. Then this subtree must have been obtained by removing a leaf joined the pendant path one of the input trees by Lemma \ref{pendant_depth_lemma}. Thus we can reconstruct an input tree by attaching a leaf at the root (resulting in an isomorphic tree), and then by Theorem \ref{unlabelled_buckets_one_from_other} recover the second input tree.
    
\end{proof}

In this theorem we consider our final case, where our two input trees have root balances $1$ and $2$. In this theorem we will primarily consider the major subtree of our root balance $1$ tree, with 3 different strategies for recovering our original trees depending on whether this major subtree has root balance $1$, $2$, or greater than $2$. In this theorem we assume that our trees have more than $7$ leaves. We will see in Theorem \ref{unlabelled_buckets_large_counterexamples} that unrecoverable pairs of trees with root balances $1$ and $2$ exist with $6$ and $7$ leaves. The strategy for recovering the original trees relies on recovering the major subtree of the major subtree of the tree with root balance $1$. If the major subtree of our root balance $1$ tree has root balance $1$ then its major subtree will have $4$ leaves when our trees have $6$, and similarly if this major subtree has root balance $2$ then its major subtree will have $4$ leaves when our trees have $7$. When this subtree has $4$ leaves we cannot recover it from its $3$-multideck, which in turn can mean that we cannot recover our original pair of trees (again due to Theorem \ref{reconstruct_from_n-1_multideck}).

\begin{Theorem}\label{unlabelled_buckets_theorem6}
    Let $T_1$ and $T_2$ be unlabelled phylogenetic trees on $n>7$ leaves with root balance $k_1$ and $k_2$ respectively. If $k_1 = 1$ and $k_2 = 2$ then $T_1$ and $T_2$ are recoverable from their $(n-1)$-bucket.
\end{Theorem}

\begin{proof}
    By Theorem \ref{unlabelled_buckets_rb_equal}, the root balances of $T_1$ and $T_2$ are recoverable.

    Let the major subtree of $T_1$ be $T_A$. By Lemma \ref{unlabelled_buckets_counting}, $B_{n-1}(\{T_1, T_2\})$ will contain $n + 1$ subtrees with root balance $1$, $n - 2$ subtrees with root balance $2$, and one subtree which could have any possible root balance. But this subtree of variable root balance is $T_A$, and so the root balance of $T_A$ is determined by $B_{n-1}(\{T_1, T_2\})$.
    
    If $T_A$ has root balance greater than 2, then it is the unique subtree in $B_{n-1}(\{T_1, T_2\})$ with that root balance, and so is uniquely determined by $B_{n-1}(\{T_1, T_2\})$. But the minor subtree of $T_1$ is a single leaf. So $T_1$ is the unique tree with $T_A$ as a major subtree, and a single leaf as a minor subtree. Thus $T_1$ is recoverable, and so by Lemma \ref{unlabelled_buckets_one_from_other}, $T_1$ and $T_2$ are recoverable from their $(n-1)$-bucket.
    
    If $T_A$ has root balance $1$ then all subtrees of $T_1$ must have root balance $1$. Thus if we consider $S$, the multiset of subtrees in $B_{n-1}(\{T_1, T_2\})$ with root balance $2$, every subtree in $S$ must be a subtree of $T_2$. 
    
    Let $T_B$ be the major subtree of $T_2$. As each subtree in $S$ must have been obtained from $T_2$ by removing a leaf from $T_2$'s major subtree,     by Lemma \ref{recover_from_root_balance_class}, $T_2$ is recoverable. Thus by Lemma \ref{unlabelled_buckets_one_from_other}, $T_1$ and $T_2$ are recoverable from their $(n-1)$-bucket.

    Now suppose $T_A$ has root balance $2$. Let $S$ denote the set of subtrees in $B_{n-1}(\{T_1, T_2\})$ with pendant depth at least $2$.
    
    Two of the subtrees in $S$ must be obtained from $T_1$ by removing a leaf from the minor subtree of $T_A$. Thus if there are only 2 subtrees in $S$ they are identical, and $T_A$ is the unique subtree obtained by adding a leaf to the minor subtree of the major subtree of one of the subtrees in $S$. Then $T_1$ is the unique tree with a single leaf as a minor subtree and $T_A$ as a major subtree. Thus $T_1$ is recoverable and so by Lemma \ref{unlabelled_buckets_one_from_other}, $T_1$ and $T_2$ are recoverable from their $(n-1)$-bucket.
    
    However, if $|S| \neq 2$, at least one subtree in $S$ must be a subtree of $T_2$. Let $T_B$ be the major subtree of $T_2$. But as the subtrees in $S$ have root balance $1$, the subtrees in $S$ which are subtrees of $T_2$ must be obtained by removing a leaf from the minor subtree of $T_2$, and so must have $T_B$ as their major subtree. But as the subtrees in $S$ have pendant depth at least 2, their major subtree has root balance $1$. Thus $T_B$ has root balance $1$. But this means that there is no $n-1$ leaf subtree of $T_2$ which both has root balance $1$ and has a major subtree with root balance $2$, as any $n-1$ leaf subtree of $T_2$ with root balance $1$ will have $T_B$ as a major subtree, and $T_B$ has root balance $1$. Thus if we consider the multiset $S'$ of subtrees in $B_{n-1}(\{T_1, T_2\})$ with root balance 1, and with a major subtree with root balance $2$, every subtree in $S'$ must be a subtree of $T_1$. But if we consider the major subtrees of the major subtrees of the subtrees in $S'$, they give the $(n-4)$ multideck of the major subtree of $T_A$, which has $n-3$ leaves. As $n > 7$, we know $n-3>4$, and by Theorem \ref{reconstruct_from_n-1_multideck} this multideck uniquely determines the major subtree of $T_A$, say $T_A'$. But $T_A$ has root balance $2$, and there is only a single tree on 2 leaves. Thus $T_A$ is the unique tree with $T_A'$ as a major subtree and the unique tree on two leaves as a minor subtree. But then $T_1$ is the unique tree with $T_A$ as a major subtree and a single leaf as a minor subtree. Thus $T_1$ is recoverable, and by Lemma \ref{unlabelled_buckets_one_from_other}, $T_1$ and $T_2$ are recoverable from their $(n-1)$-bucket.
\end{proof}

Taken together, the theorems in this section so far can be combined to show that all pairs of unlabelled trees with at least nine leaves are recoverable. Together with some simple observations about trees with smaller numbers of leaves, we can show almost all pairs of trees are recoverable, leading to a proof of our main theorem for this section, Theorem \ref{always_recoverable_unlabelled}.

\begin{proof}[Proof of Theorem \ref{always_recoverable_unlabelled}]
    If $n>8$ then Theorems \ref{unlabelled_buckets_theorem1}, \ref{unlabelled_buckets_theorem2}, \ref{unlabelled_buckets_theorem3}, \ref{unlabelled_buckets_theorem3.5}, \ref{unlabelled_buckets_theorem4}, \ref{unlabelled_buckets_theorem5}, and \ref{unlabelled_buckets_theorem6} show this holds for all possible cases for the root balances of $T_1$ and $T_2$.
    
    If $n < 4$ there is only a single tree on $n$ leaves, and so $T_1=T_2=T_3=T_4$, thus $\{T_1, T_2\}=\{T_3, T_4\}$.

    Finally, we consider the $5$-leaf unlabelled trees, depicted in Figure \ref{fig:five_leaf_unlabelled_trees_counterexamples}. Observe that there are only two unlabelled $4$-leaf trees - the caterpillar tree $C$ and the fully balanced tree $F$. Hence the $4$-multidecks of each tree are determined entirely by their counts of each.  Representing the $4$-multideck of a tree by the tuple $(c,f)$, where $c$ is the number of copies of $C$ and $f$ is the number of copies of $F$, for $5a$ we have $(5,0)$, for $5b$ we have $(1,4)$ and for $5c$ we have $(2,3)$.

    We can then easily confirm by hand that the pre-image of each tuple is unique, or doesn't exist, as seen in Table \ref{f:tuples}.
\end{proof}

\begin{table}
    \centering
    \begin{tabular}{c c}
\textbf{Tuple} & \textbf{Tree Pair} \\
\hline
$(0,10)$ & -\\
$(1,9)$ & -\\
$(2,8)$ & 5b, 5b\\
$(3,7)$ & 5b, 5c\\
$(4,6)$ & 5c, 5c\\
$(5,5)$ & -\\
$(6,4)$ & 5a, 5b\\
$(7,3)$ & 5a, 5c\\
$(8,2)$ & -\\
$(9,1)$ & -\\
$(10,0)$ & 5a, 5a\\
\end{tabular}
\caption{The pairs of trees that generate each tuple. A tuple with no pre-image is denoted by `-'.}
\label{f:tuples}
\end{table}

We now focus on when pairs of unlabelled trees are not recoverable. It is straightforward to observe that recovery is never possible with unlabelled trees with four leaves.

\begin{Theorem}
    Let $T_1$ and $T_2$ be unlabelled rooted binary phylogenetic trees on $n=4$ leaves. Then $T_1$ and $T_2$ are not recoverable from their $(n-1)$-bucket.
\end{Theorem}

\begin{proof}
    There is a unique unlabelled binary rooted phylogenetic tree on $n=3$ leaves, which we shall call $T'$. Thus, for any pair of trees $T_1$ and $T_2$, their $(n-1)$-bucket is simply eight copies of $T'$, and so no pair of trees on four leaves is recoverable. 
\end{proof}

For pairs of unlabelled trees with six, seven or eight leaves, however, the characterisation is more nuanced due to some pathological cases. We directly characterise when recovery is impossible for pairs of unlabelled trees with these numbers of leaves.

\begin{Theorem}\label{unlabelled_buckets_large_counterexamples}
    Let $T_1$ and $T_2$ be unlabelled rooted binary phylogenetic trees on $n=6,7$ or $8$ leaves. Then $T_1$ and $T_2$ are recoverable from their $(n-1)$-bucket, unless they are one of the following cases:

    \begin{enumerate}
        \item $n=6$ and $\{T_1,T_2\}$ is either the pair 6a and 6e, or the pair 6b and 6d in Figure \ref{fig:six_leaf_unlabelled_trees_counterexamples}; or
        \item $n=7$ and $\{T_1,T_2\}$ is either the pair 7a and 7d, or the pair 7b and 7c in Figure \ref{fig:seven_leaf_unlabelled_trees_counterexamples}; or
        \item $n=8$ and $\{T_1,T_2\}$ is either the pair 8a and 8c, or 2 copies of 8b in Figure \ref{fig:eight_leaf_unlabelled_trees_counterexamples}.
    \end{enumerate}
\end{Theorem}

\begin{proof}
    Let $T_1$, $T_2$, $T_3$ and $T_4$ be rooted binary phylogenetic trees on $n$ leaves such that $B_{n-1}\{T_1, T_2\}) = B_{n-1}\{T_3, T_4\})$, and $\{T_1, T_2\} \neq \{T_3, T_4\}$.
    
    Suppose $n=6$. Then by Theorems \ref{unlabelled_buckets_theorem1}, \ref{unlabelled_buckets_theorem2}, \ref{unlabelled_buckets_theorem4}, and \ref{unlabelled_buckets_theorem5}, the sets of root balances of $T_1$ and $T_2$, and $T_3$ and $T_4$ are either $\{3, 3\}$ or $\{1, 2\}$. As there is only a single tree on 6 leaves with root balance $3$, contradicting that the trees are distinct, this set must be $\{1, 2\}$. Suppose without loss of generality that $T_1$ and $T_3$ have root balance $1$, and $T_2$ and $T_4$ have root balance $2$. There are only 2 distinct trees on 6 leaves with root balance $2$, so say $T_2$ is Tree 6e shown in Figure \ref{fig:six_leaf_unlabelled_trees_counterexamples}, and $T_4$ is Tree 6d shown in Figure \ref{fig:six_leaf_unlabelled_trees_counterexamples}. If we examine Table \ref{table:six_leaf_subtrees}, we see that the only possible choices to make the buckets equal are if $T_1$ is Tree 6a shown in Figure \ref{fig:six_leaf_unlabelled_trees_counterexamples}, and $T_3$ is Tree 6b shown in Figure \ref{fig:six_leaf_unlabelled_trees_counterexamples}.
    Suppose $n=7$. Then by Theorems \ref{unlabelled_buckets_theorem1}, \ref{unlabelled_buckets_theorem2}, \ref{unlabelled_buckets_theorem4}, and \ref{unlabelled_buckets_theorem5}, the sets of root balances of $T_1$ and $T_2$, and $T_3$ and $T_4$ are either $\{3, 3\}$ or $\{1, 2\}$. There are only 2 distinct trees on 7 leaves with root balance $3$, say $A$ and $B$. Then by Lemma \ref{unlabelled_buckets_one_from_other}, for $\{T_1, T_2\} \neq \{T_3, T_4\}$ to hold, $T_1 = T_2 = A$, and $T_3 = T_4 = B$. But this contradicts Theorem \ref{reconstruct_from_n-1_multideck}, so sets of root balances of $T_1$ and $T_2$, and $T_3$ and $T_4$ must be $\{1, 2\}$. Now consider the trees in $B_{n-1}\{T_1, T_2\})$ with root balance $2$. We expect $7-2=5$ of these to be from the root balance $2$ tree, so if there are exactly 5, all come from that tree. If we consider the major subtrees of these, they form the $4$-multideck of the major subtree of the root balance $2$ tree, and thus we can recover that tree, and by \ref{unlabelled_buckets_one_from_other}, $T_1$ and $T_2$ are recoverable. Thus we must have an extra root balance $2$ subtree, which can only occur when the root balance $1$ subtree has a root balance $2$ major subtree. The trees on 7 leaves with root balance $1$ and a major subtree with root balance $2$ are Trees 7a and 7b shown in Figure \ref{fig:seven_leaf_unlabelled_trees_counterexamples}. The trees on 7 leaves with root balance $2$ are Trees 7c, 7d, and 7e shown in Figure \ref{fig:seven_leaf_unlabelled_trees_counterexamples}. The 6-multidecks of these trees are shown in Table \ref{table:seven_leaf_subtrees}. In this table we can see that no combination of 2 trees that do not include 7e will have 3 copies of 6e in their 6-bucket, and so if either $T_1$ or $T_2$ is 7e, $T_1$ and $T_2$ are recoverable from their 6-bucket. If we examine the multidecks of the remaining trees, shown in Table \ref{table:seven_leaf_subtrees}, we find that the only combination for which the pairs of trees have identical buckets is if $\{T_1, T_2\}$ are Trees 7a and 7d, and $\{T_3, T_4\}$ are Trees 7b and 7c, as shown in Figure \ref{fig:seven_leaf_unlabelled_trees_counterexamples}.
    Suppose $n=8$. Then by Theorems \ref{unlabelled_buckets_theorem1}, \ref{unlabelled_buckets_theorem2}, \ref{unlabelled_buckets_theorem4}, \ref{unlabelled_buckets_theorem5}, and \ref{unlabelled_buckets_theorem6}, the sets of root balances of $T_1$ and $T_2$, and $T_3$ and $T_4$ are $\{4, 4\}$. The 3 distinct rooted trees on $8$ leaves are Tree 8a, 8b and 8c shown in Figure \ref{fig:eight_leaf_unlabelled_trees_counterexamples}. The $7$-multidecks of these trees are 8 copies of Tree 7f shown in Figure \ref{fig:seven_leaf_unlabelled_trees_counterexamples} for 8a, 8 copies of Tree 7g shown in Figure \ref{fig:seven_leaf_unlabelled_trees_counterexamples} for 8c, and 4 each of 7f and 7g for 8b. The only possible assignment of these trees to $T_1$, $T_2$, $T_3$ and $T_4$ for which $B_{n-1}\{T_1, T_2\}) = B_{n-1}\{T_3, T_4\})$, and $\{T_1, T_2\} \neq \{T_3, T_4\}$ is if $T_1$ is 8a and $T_2$ is 8c; and $T_3$ and $T_4$ are both 8b (or vice versa). Thus $T_1$ and $T_2$ are recoverable unless $\{T_1,T_2\}$ is either the pair 8a and 8c, or 2 copies of 8b.
\end{proof}

\begin{table}
    \begin{center}
        \begin{tabular}{c||c|c|c|}
             & 5a & 5b & 5c \\
             \hline \hline
             6a & 6 & & \\
             \hline
             6b & 4 & 2 & \\
             \hline
             6c & 2 & 3 & 1 \\
             \hline \hline
             6d & 2 & & 4 \\
             \hline
             6e & & 2 & 4 \\
             \hline
        \end{tabular}
    \end{center}
    \caption{The 5-multidecks of the 6 leaf trees with root balance other than $3$. Trees in the upper section of the table have root balance $1$, and those below have root balance $2$. Labels refer to trees shown in Figures \ref{fig:five_leaf_unlabelled_trees_counterexamples} and \ref{fig:six_leaf_unlabelled_trees_counterexamples}.}
    \label{table:six_leaf_subtrees}
\end{table}

\begin{table}
    \begin{center}
        \begin{tabular}{c||c|c|c|c|c|}
             & 6a & 6b & 6c & 6d & 6e \\
             \hline \hline
             7a & 2 & & 4 & 1 & \\
             \hline
             7b & & 2 & 4 & & 1 \\
             \hline \hline
             7c & 2 & & & 5 & \\
             \hline
             7d & & 2 & & 4 & 1 \\
             \hline
             7e & & & 2 & 2 & 3 \\
             \hline
        \end{tabular}
    \end{center}
    \caption{The 6-multidecks of the 7 leaf trees with root balance $1$ and a major subtree with root balance $2$, and those with root balance $2$. Trees in the upper section of the table have root balance $1$, and those below have root balance $2$. Labels refer to trees shown in Figures \ref{fig:six_leaf_unlabelled_trees_counterexamples} and \ref{fig:seven_leaf_unlabelled_trees_counterexamples}.}
    \label{table:seven_leaf_subtrees}
\end{table}

\begin{figure}[ht]
    \centering
    \begin{subfigure}[b]{0.3\textwidth}
        \centering
        \includegraphics[width=\textwidth]{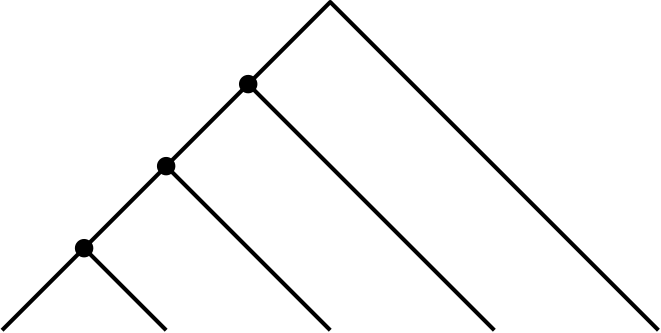}
        \caption*{Tree 5a}
    \end{subfigure}
    \hfill
    \begin{subfigure}[b]{0.3\textwidth}
        \centering
        \includegraphics[width=\textwidth]{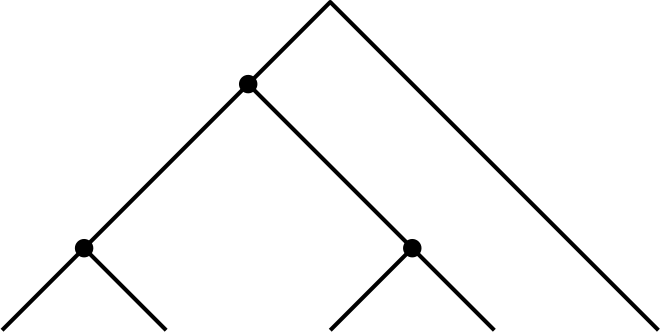}
        \caption*{Tree 5b}
    \end{subfigure}
    \hfill
    \begin{subfigure}[b]{0.3\textwidth}
        \centering
        \includegraphics[width=\textwidth]{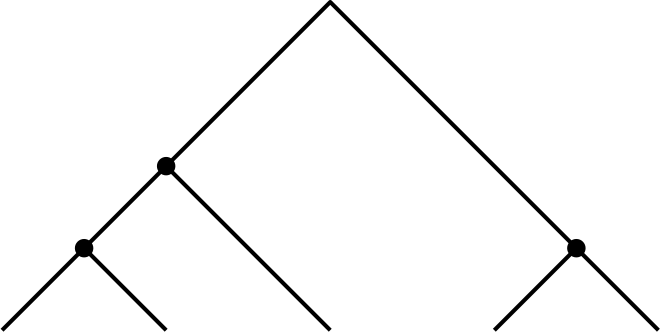}
        \caption*{Tree 5c}
    \end{subfigure}
    \caption{Trees on 5 leaves}
    \label{fig:five_leaf_unlabelled_trees_counterexamples}
\end{figure}

\begin{figure}[ht]
    \centering
    \begin{subfigure}[b]{0.3\textwidth}
        \centering
        \includegraphics[width=\textwidth]{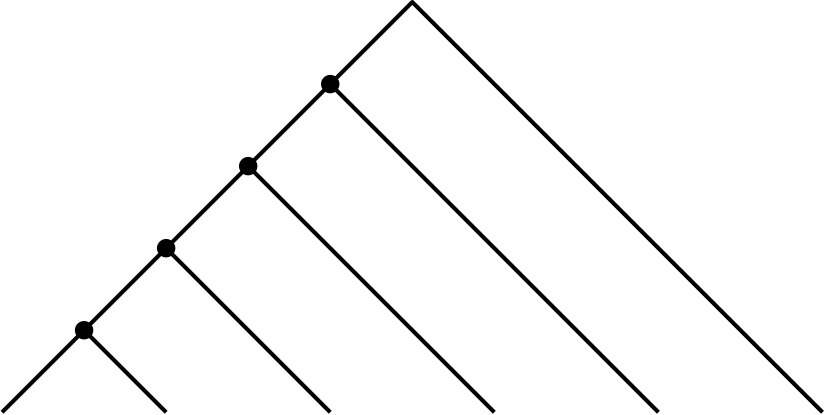}
        \caption*{Tree 6a}
    \end{subfigure}
    \hfill
    \begin{subfigure}[b]{0.3\textwidth}
        \centering
        \includegraphics[width=\textwidth]{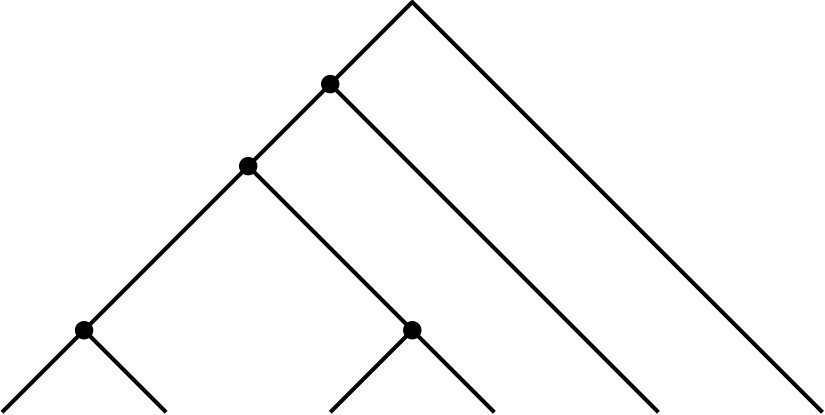}
        \caption*{Tree 6b}
    \end{subfigure}
    \hfill
    \begin{subfigure}[b]{0.3\textwidth}
        \centering
        \includegraphics[width=\textwidth]{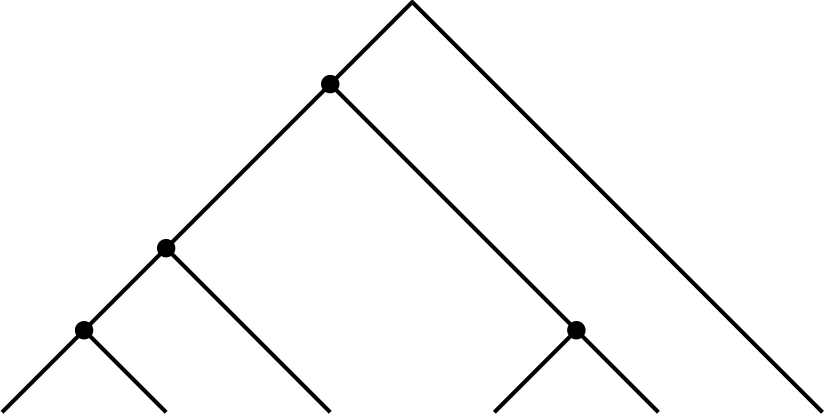}
        \caption*{Tree 6c}
    \end{subfigure}
    \hfill
    \begin{subfigure}[b]{0.3\textwidth}
        \centering
        \includegraphics[width=\textwidth]{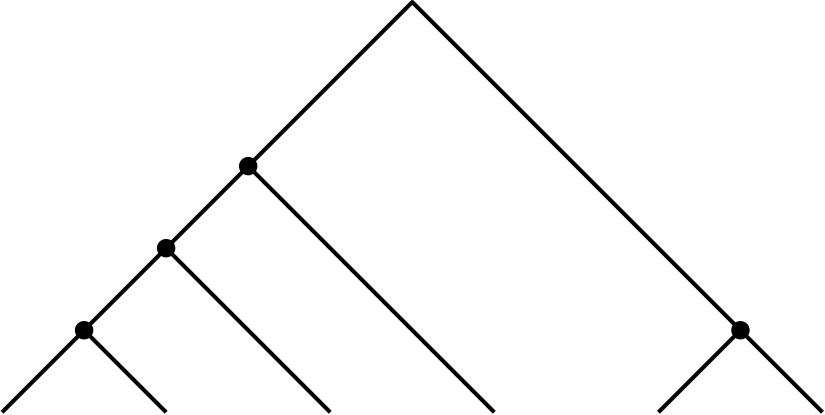}
        \caption*{Tree 6d}
    \end{subfigure}
    \hfill
    \begin{subfigure}[b]{0.3\textwidth}
        \centering
        \includegraphics[width=\textwidth]{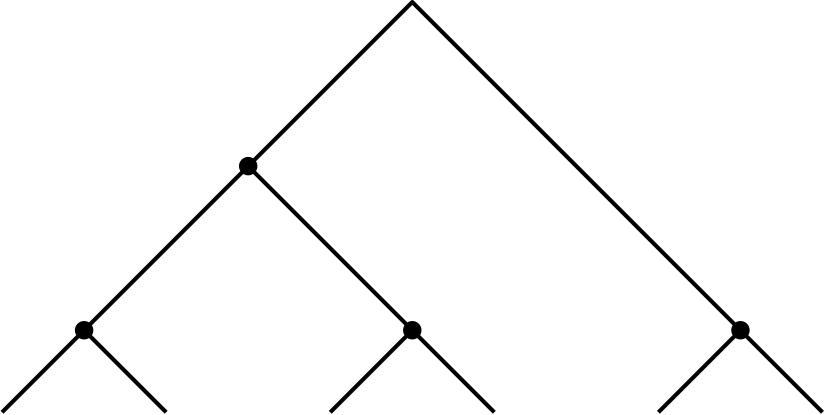}
        \caption*{Tree 6e}
    \end{subfigure}
    \caption{Trees on 6 leaves}
    \label{fig:six_leaf_unlabelled_trees_counterexamples}
\end{figure}

\begin{figure}[ht]
    \centering
    \begin{subfigure}[b]{0.225\textwidth}
        \centering
        \includegraphics[width=\textwidth]{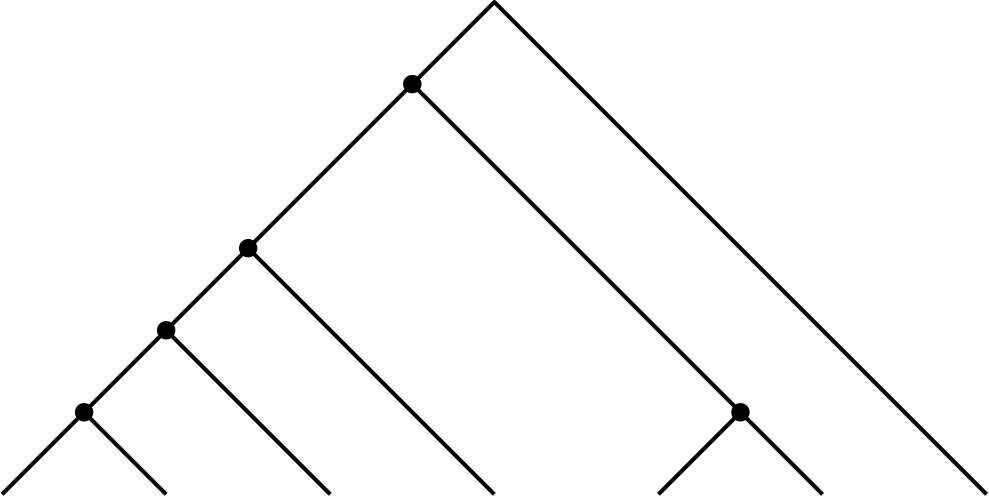}
        \caption*{Tree 7a}
    \end{subfigure}
    \hfill
    \begin{subfigure}[b]{0.225\textwidth}
        \centering
        \includegraphics[width=\textwidth]{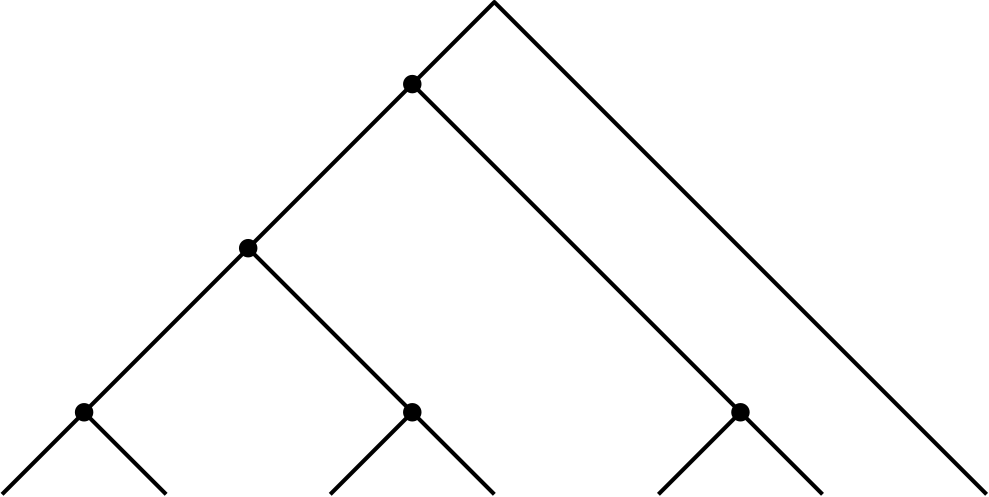}
        \caption*{Tree 7b}
    \end{subfigure}
    \hfill
    \begin{subfigure}[b]{0.225\textwidth}
        \centering
        \includegraphics[width=\textwidth]{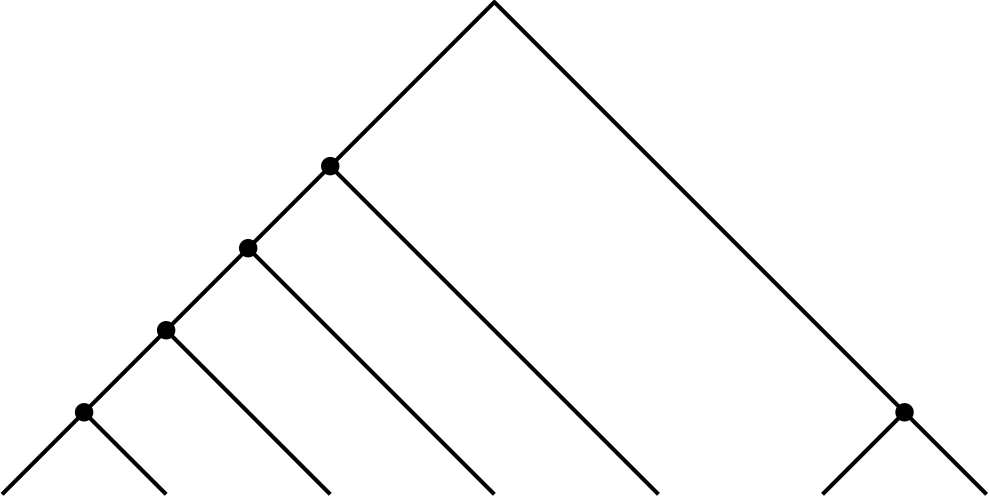}
        \caption*{Tree 7c}
    \end{subfigure}
    \hfill
    \begin{subfigure}[b]{0.225\textwidth}
        \centering
        \includegraphics[width=\textwidth]{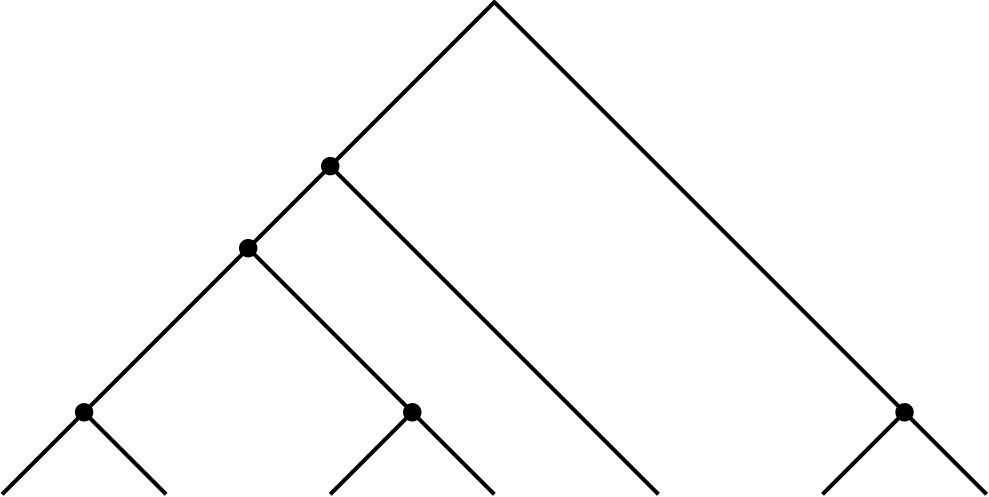}
        \caption*{Tree 7d}
    \end{subfigure}
    \hfill
    \begin{subfigure}[b]{0.225\textwidth}
        \centering
        \includegraphics[width=\textwidth]{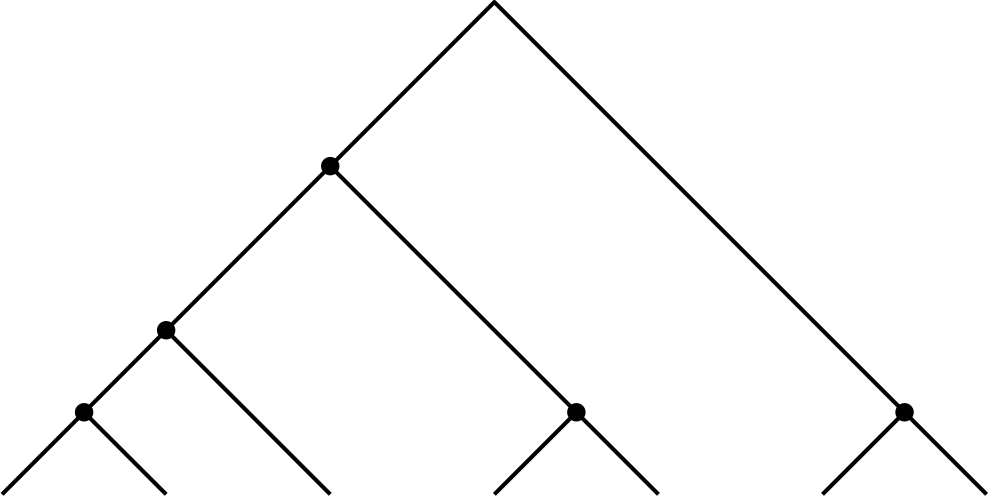}
        \caption*{Tree 7e}
    \end{subfigure}
    \hfill
    \begin{subfigure}[b]{0.225\textwidth}
        \centering
        \includegraphics[width=\textwidth]{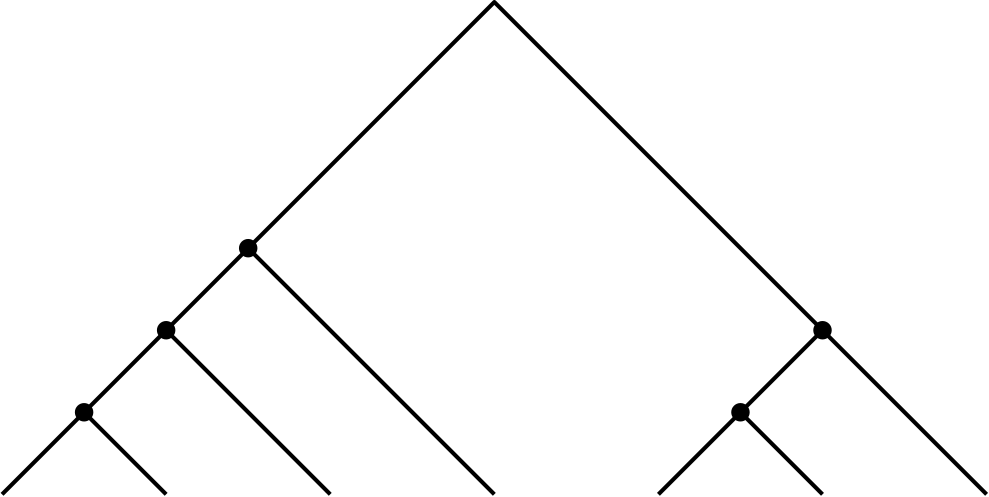}
        \caption*{Tree 7f}
    \end{subfigure}
    \hfill
    \begin{subfigure}[b]{0.225\textwidth}
        \centering
        \includegraphics[width=\textwidth]{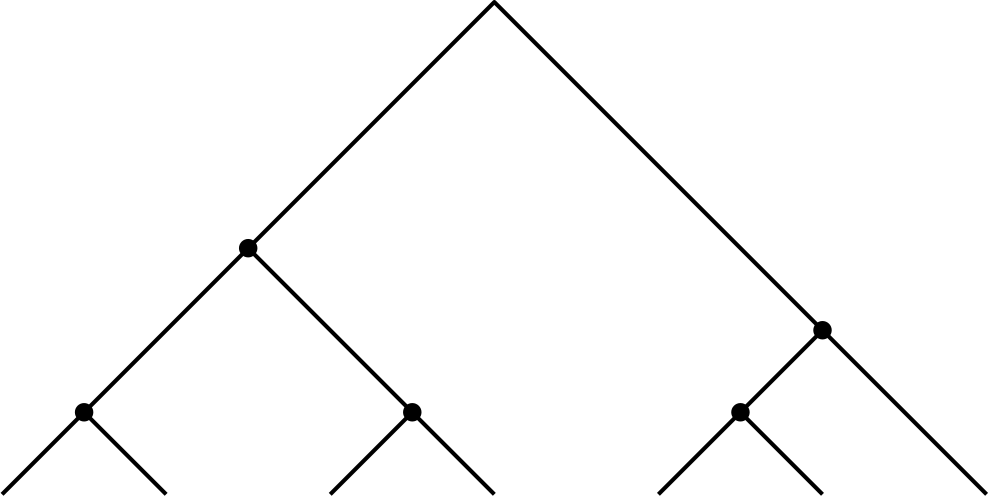}
        \caption*{Tree 7g}
    \end{subfigure}
    \caption{Trees on 7 leaves}
    \label{fig:seven_leaf_unlabelled_trees_counterexamples}
\end{figure}

\begin{figure}[ht]
    \centering
    \begin{subfigure}[b]{0.3\textwidth}
        \centering
        \includegraphics[width=\textwidth]{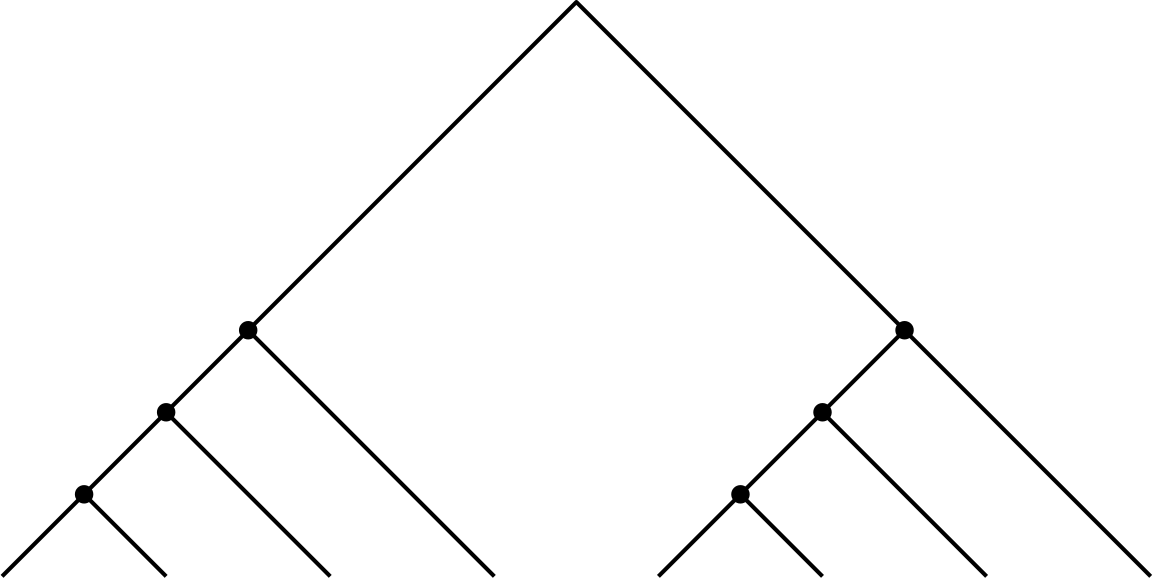}
        \caption*{Tree 8a}
    \end{subfigure}
    \hfill
    \begin{subfigure}[b]{0.3\textwidth}
        \centering
        \includegraphics[width=\textwidth]{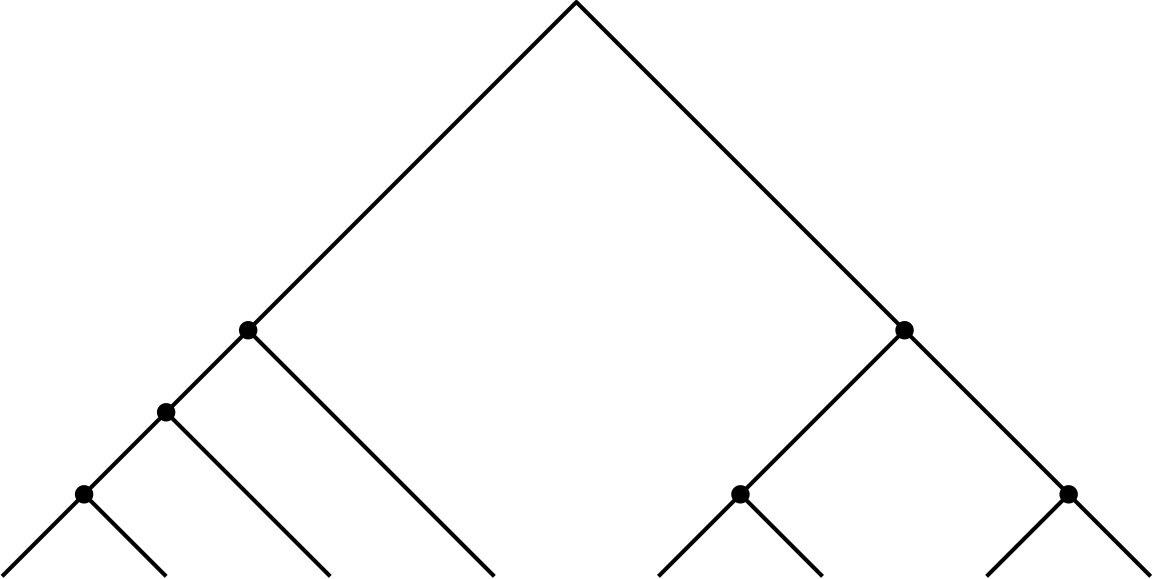}
        \caption*{Tree 8b}
    \end{subfigure}
    \hfill
    \begin{subfigure}[b]{0.3\textwidth}
        \centering
        \includegraphics[width=\textwidth]{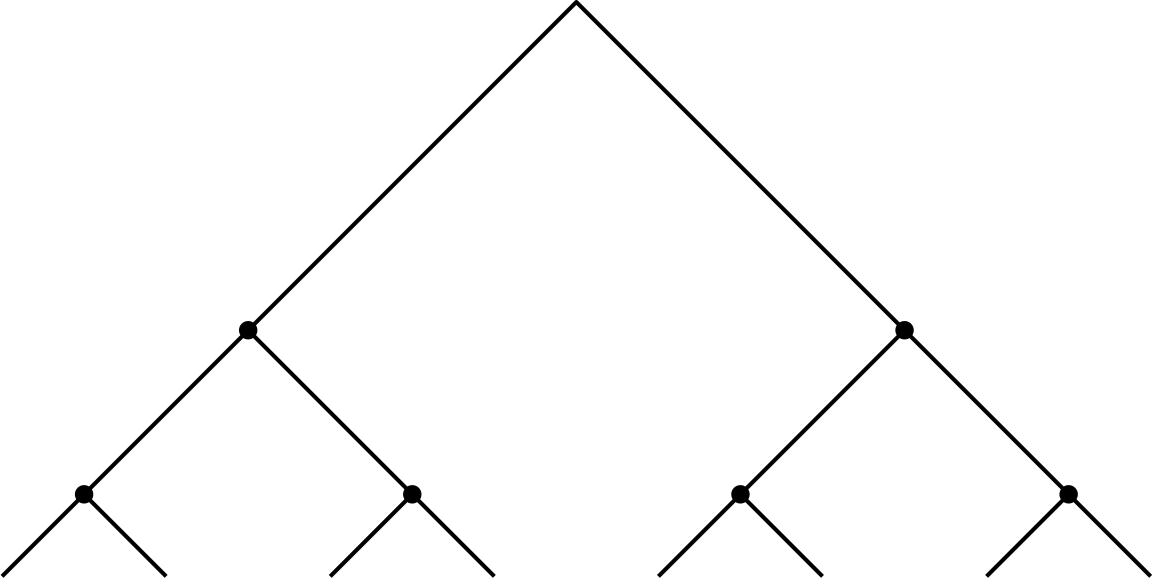}
        \caption*{Tree 8c}
    \end{subfigure}
    \caption{Trees on 8 leaves}
    \label{fig:eight_leaf_unlabelled_trees_counterexamples}
\end{figure}

Given that recoverability from $(n-1)$-buckets in the unlabelled case is always possible outside a few pathological counterexamples, one may be tempted to think that the same may be true in the labelled case. This is in fact true, as we will see in Section \ref{s:labeln-1}. However, if we consider $k$-buckets with $k < n-1$, this is not necessarily the case, as can be seen from the examples in Figure \ref{fig:five_leaf_labelled_counterexample_1}. We note that, in this example, one pair of trees can be constructed from the other by `swapping' the subtrees obtained by restricting to $Y = \{a,b,c\}$ between the two trees. This phenomenon can be generalised, leading to an infinite class of counterexamples, as we will see in the next section.

We therefore will consider the labelled cases in the next two sections. We will first consider the case of labelled trees and their $3$-buckets, as some results that will be derived are then useful for $k$-buckets with larger $k$.

\section{Labelled trees and \texorpdfstring{$3$}{3}-buckets}\label{s:label3}

Again, a priori, it is not clear whether the decrease in information caused by considering fewer leaves may be counteracted by the potential increase in information by considering labelled trees (or indeed, whether labelling makes the recovery process more difficult). As we will see, in the case that $k=3$, recovery is considerably more difficult than the $(n-1)$-bucket unlabelled problem, as there is an infinite class of unrecoverable pairs. We will first show that there is such an infinite class, generated by swapping subtrees with the same support between the two input trees.

\begin{Theorem}
    Let $T_1$ and $T_2$ be labelled phylogenetic trees on $X$. If there is some subset $A$ of $X$ and vertices $x_1$ and $x_2$ of $T_1$ and $T_2$ respectively such that the set of leaves that are descendants of $x_1$ and $x_2$ are exactly $A$, then the trees $T_3$ and $T_4$ obtained by exchanging the subtrees rooted at $x_1$ and $x_2$ have the same multiset of rooted triples as $T_1$ and $T_2$.
\end{Theorem}

\begin{proof}
    Consider some subset $\{a, b, c\}$ of $X$. Either $0$, $1$, $2$ or $3$ of $a$, $b$, and $c$ are in $A$.
    
    If $a, b, c, \notin A$, then swapping these subtrees does not change the triples containing $a$, $b$ and $c$. Thus the triple of $a$, $b$ and $c$ in $T_1$ is also a triple in $T_3$, and the same holds of $T_2$ and $T_4$.
    
    If $a \in A$ and $b, c \notin A$ then swapping these subtrees does not change the position of $a$ in relation to $b$ and $c$. Thus the triple of $a$, $b$ and $c$ in $T_1$ is also a triple in $T_3$, and the same holds of $T_2$ and $T_4$.
    
    If $a, b \in A$, and $c \notin A$, then $ab|c$ must be a triple of all of $T_1$, $T_2$, $T_3$, and $T_4$.
    
    If $a, b, c \in A$. Then the triple of $a$, $b$ and $c$ in $T_1$ must also be a triple in $T_4$, as both contain the same subtree containing $a$, $b$ and $c$, and the same holds for $T_2$ and $T_3$.
\end{proof}

Note that if the subtrees rooted at $x_1$ and $x_2$ are the same, or $T_1$ and $T_2$ differ only in these subtrees, then $\{T_1, T_2\} = \{T_3, T_4\}$. Additionally, if $|A|$ is $1$ or $2$, then there is only a single phylogenetic tree on $A$, and so these subtrees must be the same. Similarly if $n - |A|$ is $1$, then the parts of $T_1$ and $T_2$ not in these subtrees must only be a single leaf, and so $T_1$ and $T_2$ will differ only in these subtrees. Therefore if $|A|=1$ or $2$, or $n-|A|=1$, this subtree swap will not result in a distinct pair of trees.

\begin{Corollary}
    For each $n>4$ there exist distinct labelled phylogenetic trees $T_1$, $T_2$, $T_3$, and $T_4$ on $X$ with $|X|=n$ such that $B_3(\{T_1, T_2\}) = B_3(\{T_3, T_4\})$, and hence these pairs are not recoverable from their $3$-bucket.
\end{Corollary}

We will proceed to show that subtree swaps are in fact the only way in which to construct multiple pairs of trees with the same $3$-buckets, but this will require substantial additional machinery. 

We first define a graph designed to recover maximal clusters in a tree, used in the well-known BUILD algorithm \cite{aho1981inferring}. This will be adapted for use in $3$-bucket recovery.

\begin{Definition}
    Let $R$ be a set of triples, and $Y \subseteq X$. Define the \emph{BUILD graph of $R$ restricted to $Y$}, denoted $BG(R,Y)$, to be the graph with vertex set $Y$, and an edge between vertices $a$ and $b$ if and only if there is some triple $r \in R$ of the form $ab|c$ such that $\supp(r) \subseteq Y$.
\end{Definition}

We must also define the usual concept of a cluster of a tree.

\begin{Definition}
    Let $T$ be a labelled phylogenetic tree on $X$, and $v$ some vertex of $T$. Let $Y$ be the set of leaves descended from $v$. Then $Y$ is referred to as a \emph{cluster} of $T$. If two sets $A$ and $B$ have the property $A \cap B = \emptyset, A \subseteq B$ or $B \subseteq A$, then these sets are said to be \emph{cluster-compatible}.
\end{Definition}

Clusters have two well-known properties that we shall use here.

\begin{Theorem}\label{t:compatible.clusters}
    A collection $H$ of subsets of $X$ is the cluster set of a labelled phylogenetic tree $T$ if and only if for each $A,B \in H$, $A$ and $B$ are cluster-compatible.
\end{Theorem}

\begin{Theorem}
    A subset $Y$ of $X$ is a cluster of the labelled phylogenetic tree $T$ if and only if for all $a,b \in Y$ and $c \in X \backslash Y$, $T$ displays the triple $ab|c$.
\end{Theorem}

\begin{Definition}
    Let $G$ be a graph with vertex set $X$ such that there is some $A \sqcup B = X$ for which $G$ is the disjoint union of the complete graphs on $A$ and $B$, and both $A$ and $B$ are non-empty. Then we call $G$ a \emph{disclique of $A$ and $B$}, or just a \emph{disclique} if it is a disclique of $A$ and $B$ for some $A$ and $B$. 
\end{Definition}

The following theorem is due to \cite{aho1981inferring}, but rewritten with our terminology.

\begin{Theorem}[\cite{aho1981inferring}]\label{t:buildistwocliques}
Let $R$ be the $3$-multideck of a labelled phylogenetic tree $T$ on $X$, and let $Y \subseteq X$. Let the major and minor subtrees of $T|_Y$ have leaf sets $A$ and $B$ respectively (noting $A \sqcup B = Y$).  Then $BG(R,Y)$ is a disclique of $A$ and $B$.
\end{Theorem}

As a consequence of Theorem \ref{t:buildistwocliques}, if $R$ is the $3$-bucket of $T_1$ and $T_2$, $BG(R,Y)$ is the union of two pairs of disjoint cliques that cover the vertices. We will now show that we can uniquely recover a pair of discliques from their graph union.

\begin{Theorem}
    Every union of discliques on a graph with at least two vertices admits a unique decomposition as the union of two discliques
\end{Theorem} 

\begin{proof}
    The theorem is trivially true in the case of the graph being two vertices, as both discliques must simply be the unique pair of size one cliques. Hence, we can suppose the graph has at least $3$ vertices, and hence at least one edge.

    We will first show that any maximal clique in $G$ containing some vertex $a$ must be one of the two disclique parts containing $a$.
    
    Let $G$ be a union of discliques, and select some vertex $a$ in $G$. Let one disclique be $A \sqcup B$ and the other $C \sqcup D$, and without loss of generality, suppose $a \in A$ and $a \in C$ (hence $a \in A \cap C$). We now consider any maximal clique $K$ containing $a$, and we claim that either $K=A$ or $K=C$. 

    We now partition the vertices in the graph into the four parts $A \cap C, A \cap D, B \cap C$ and $B \cap D$. We note that $a$ can only be adjacent to vertices that are in either $A$ or $C$, and is not adjacent to any vertex in $B \cap D$. It follows that $K \subseteq A \cup C$.

    We now show that $K$ cannot contain both a vertex in $A \cap D$ and $B \cap C$ simultaneously. Seeking a contradiction, suppose otherwise, and let $x \in A \cap D$ and $y \in B \cap C$. We can then see that $x$ and $y$ are not adjacent in the disclique $A \sqcup B$, as $x \in A$ and $y \in B$. Furthermore, they are not adjacent in the disclique $C \sqcup D$, as $x \in D$ and $y \in C$. Hence $x$ and $y$ are not adjacent, contradicting that $K$ is a clique.

    It follows that either $K \subseteq (A\cap C) \cup (A \cap D) = A$ or $K \subseteq (A \cap C) \cup (B \cap C) = C$. Finally, since $K$ is maximal, it follows that either $K = A$ or $K= C$ as claimed.

    It follows that we can immediately recover at least one disclique, by finding some maximal clique, say $A$, and by taking the complement we recover $B$ too.

    Now, if there are no edges between a vertex in $A$ and a vertex in $B$, it follows that $C = A$ and $B = D$, or vice versa. This can be seen by considering the cases - either $C \cap A \ne \emptyset$ or $D \cap A \ne \emptyset$, and without loss of generality suppose the former. If $C \subset A$, then $D$ must contain $B$ and some vertex in $A$, and hence there would be an edge from a vertex in $B$ to a vertex in $A$, a contradiction. If $C$ also contains some vertex outside of $A$, then it must intersect $B$, and since $C$ is a clique again there must be an edge from $B$ to $A$. The only remaining possibility is thus that $C=A$. Thus, in this case, we can immediately uniquely decompose $G$ into the unique decomposition of two copies of the $A \sqcup B$.

    Otherwise, there exists some edge $e$ between $A$ and $B$. Select either vertex in $e=\{u,v\}$, say $u$, and consider a maximal clique containing $u$ that also contains $e$. As previously shown, any maximal clique in $G$ containing the vertex $u$ must be one of the two disclique parts containing $u$, and it is not $A$ or $B$ as it contains $e$, so must be $C$ or $D$. Again, we can then recover the other part of this disclique by taking the complement.

    Hence we can recover $A \sqcup B$ and $C \sqcup D$, and the theorem is proven.
\end{proof}

We proceed by describing an algorithm, Algorithm \ref{alg:recover-clusters}, for recovering the multiset of clusters, as well as information about their descendant relation. We will recursively construct a Hasse diagram of clusters, with the relation being inclusion. The key observation is that for each currently identified cluster $C$, the rooted triples whose support lies in $C$ determine the unordered pair of induced subtrees of $T_1|_C$ and $T_2|_C$, obtained via the disclique composition. Much of the remainder of the work is then identifying which subtree pair belongs to which tree. At each stage, every already discovered cluster is correctly identified, its multiplicity is known, and its children in each input tree are either already determined or are recoverable from the disclique decomposition.

\begin{algorithm}[ht]
\caption{Recovering the Hasse diagram of clusters}
\label{alg:recover-clusters}
\begin{algorithmic}[1]
\State Initialize the Hasse diagram $\mathcal{H}$ with the cluster $X$.
\Statex
\While{there are clusters whose multiplicity has not been assigned}
    \State Choose a maximal cluster $C$ in $\mathcal{H}$ by inclusion among those with out-degree $0$.
    \If{$|C|=1$}
        \State Assign multiplicity $2$ to $C$.
    \Else
        \State Perform the disclique decomposition of $C$.
        \State Look for a cluster $C'$ in $\mathcal{H}$ with children $C'_1$ and $C'_2$
        \Statex \hspace{\algorithmicindent}such that $C'_1,C'_2$ partition $C'$, and
        \Statex \hspace{\algorithmicindent}$C \nsubseteq C'_1$ and $C \nsubseteq C'_2$.
        \If{no such cluster $C'$ exists}
            \State Add the $2$ or $4$ unique cliques from the disclique decomposition
            \Statex \hspace{\algorithmicindent} to $\mathcal{H}$ as children of $C$.
            \State Assign multiplicity $2$ to $C$.
        \Else
            \State Restrict $C'_1$ and $C'_2$ to $C$ to obtain a disclique.
            \State Discard this disclique from the pair obtained from $C$.
            \State Add the remaining pair of cliques to $\mathcal{H}$ as children of $C$.
            \State Assign multiplicity $1$ to $C$.
        \EndIf
    \EndIf
\EndWhile
\end{algorithmic}
\end{algorithm}

\begin{Theorem}\label{t:hasse.union}
    The Hasse diagram obtained by Algorithm \ref{alg:recover-clusters} is the union of the Hasse diagrams of the clusters of $T_1$ and $T_2$ by inclusion, together with a function mapping each cluster to its multiplicity. 
\end{Theorem}

\begin{proof}
  Certainly $X$ is a cluster in both trees. Proceeding by induction, suppose we have obtained a cluster $C$ of some input tree, without loss of generality say $T_1$, in a previous step of the algorithm, and $C$ is maximal with respect to inclusion of clusters obtained thus far in the algorithm. Note that if $|C|=1$ then it contains a single leaf, which is always a cluster of both trees and is assigned multiplicity two. Otherwise, assume $|C|>1$.
  
  Using the disclique decomposition, we obtain the multiset of major and minor clusters of $T_1|_C$ and $T_2|_C$. 
  
  Suppose $C$ is a cluster of $T_2$ in addition to $T_1$. Then there exists no obtained $C'$ with children (in the Hasse diagram) $C'_1$ and $C'_2$ such that 
  
  $C$ is a subset of neither $C'_1$ or $C'_2$, as otherwise $C'_1$ and $C'_2$ would not be cluster-compatible with $C$ (by Theorem \ref{t:compatible.clusters}), and therefore could not be clusters of either tree. Additionally, in particular, $T_1|_C$ and $T_2|_C$ are precisely the subtrees of $T_1$ and $T_2$ rooted at the vertex corresponding to $C$, and thus must be clusters of some input tree. Additionally, since $C$ appears in both trees, it must have multiplicity two and thus the algorithm correctly identifies its multiplicity. 

  Now we suppose $C$ is a cluster of only one input tree, say $T_1$. In this case, we claim there exists some obtained $C'$ (in $T_2$) with children (in the Hasse diagram) $C'_1$ and $C'_2$ such that 
  
  $C$ is a subset of neither $C'_1$ or $C'_2$. Let $D$ be the minimal cluster containing $C$ that has multiplicity $2$. The cluster $D$ must have four children in the Hasse diagram, as if it only had two its children are present in both trees. One of these children contains $C$, call it $D_1$, and has a corresponding cluster $D_2$ such that $D_1 \sqcup D_2 = D$. Consider a minimal descendant $M$, by inclusion, of the other two children of $D$ that is a superset of $C$, including themselves (if none exist, then both the other two children must have non-trivial intersection with $C$, and we have found our $C'_1$ and $C'_2$). Otherwise, of the two to four descendants of $M$ (which must have descendants, as $M$ must be considered before $C$ in the algorithm), must have a pair which union to be $M$, none of which are $C$ as it is not present in both trees and neither can be a superset, so must both have non-trivial intersection and we again have found our $C'_1$ and $C'_2$.

  It follows that the restriction of $C'_1$ and $C'_2$ to $C$ are the clusters of the major and minor subtrees of $T_2|_C$, so the other pair of clusters must be the clusters of $T_1$. Additionally, since $C$ appears in only one tree, it must have multiplicity one and thus the algorithm correctly identifies its multiplicity. 

  Finally, we note that each introduction of an edge $(C_1,C_2)$ in the Hasse diagram indicates that $C_2$ is a subcluster of $C_1$ in some input tree, and the theorem follows.  
\end{proof}

This implies the following useful corollary. 

\begin{Corollary}\label{labelled_triples_same_clusters}
    Let $T_1$, and $T_2$ be labelled phylogenetic trees on $X$. Then the multiset of clusters of $T_1$ and $T_2$ is recoverable from their $3$-bucket.
\end{Corollary}

By considering the $2$ element clusters we obtain the following additional Corollary.

\begin{Corollary}\label{labelled_triples_same_cherries}
    Let $T_1$, and $T_2$ be labelled phylogenetic trees on $X$. Then the multiset of cherries of $T_1$ and $T_2$ is recoverable.
\end{Corollary}

\begin{Lemma}\label{l:split.clusterset}
    Let $T_1$ and $T_2$ be labelled phylogenetic trees on $X$ (with $|X|>2$) with cluster sets $C(T_1)$ and $C(T_2)$ such that $C(T_1) \cap C(T_2)$ consists only of $X$ and the singleton sets. Then $T_1$ and $T_2$ are recoverable from their $3$-bucket. 
\end{Lemma}

\begin{proof}
    By Theorem \ref{t:hasse.union}, the union of the Hasse diagrams is recoverable. We now consider the linked pairs of clusters descended from $X$ in the Hasse diagram (i.e. the pair $A$ and $B$, which partition $X$, and the pair $C$ and $D$, which partition $X$). Each of these clusters must be distinct - if some cluster is not a singleton, then by the assumption that $C(T_1) \cap C(T_2)$ consists only of $X$ and the singleton sets, or if a cluster, say $A$, is a singleton, then by the fact that this forces $B=X \backslash A$, which is not a singleton, and so neither $C$ nor $D$ can be equal to $A$ as this would force the other to be $X \backslash A$ again, contradicting the $C(T_1) \cap C(T_2)$ assumption again. 

    We then consider any cluster $M$ in the diagram that is not a singleton or $X$. We claim that $M$ must have in-degree $1$. Certainly this is true for $A,B,C$ and $D$. Otherwise, suppose $M$ has in-degree greater than one. The cluster $M$ may only obtain a maximum in-degree of $1$ from a given tree. It follows that therefore $M$ must be a cluster of both trees, which contradicts the assumptions of the theorem. 
    
    Therefore every cluster in the Hasse diagram belongs to a unique tree. Consequently, once the tree containing a cluster $C$ has been
    identified, every non-singleton child of $C$ must belong to the same tree, as an edge from $C$ to such a child is an edge in the Hasse diagram of that tree. Starting with the two root splits identified above, this determines inductively which tree contains every proper non-singleton cluster. Singleton clusters may occur in both trees, but this causes no ambiguity, since their parent has already been assigned to a tree.

    Thus, the Hasse diagram union has a unique decomposition, up to exchanging $T_1$ and $T_2$, into the two Hasse diagrams of $T_1$ and $T_2$. Hence $T_1$ and $T_2$ are recoverable from their $3$-bucket.
\end{proof}

We are now finally able to prove the main result of this section, precisely classifying recoverability of pairs of trees from their $3$-buckets up to subtree swaps.

\begin{Theorem}\label{t:full_subtree_swap}
    Let $T_1$ and $T_2$ be labelled phylogenetic trees that have a $k$-bucket from which the multiset of clusters of the input trees is recoverable. Then $T_1$ and $T_2$ are recoverable up to a finite sequence of subtree swaps.

    In particular, a pair of binary phylogenetic trees $T_1$ and $T_2$ is always recoverable from their $3$-bucket, up to a finite sequence of subtree swaps.
\end{Theorem}

\begin{proof}
We prove this claim via induction on the number of leaves. Certainly the claim is true for $n<4$, so we suppose that the claim is true for $n-1$ leaves. Let $A_1$ denote a minimal subcluster which appears in both $T_1$ and $T_2$ and is not a singleton of $X$. If no such $A$ exists then Lemma \ref{l:split.clusterset} applies and $T_1$ and $T_2$ are recoverable. 

Otherwise we then perform a quotient, in which we delete the subclusters of $A_1$, then each other cluster $C$ in the multiset for which $A_1 \subseteq C$ is replaced by $C \backslash A_1 \cup \{a_1\}$, where $a_1$ is a new leaf. Contracting a common cluster $A_1$ in both trees to a new leaf $a$ sends the multiset of clusters of the pair to the corresponding multiset of clusters of the contracted pair, with the clusters contained in $A_1$ removed. This results in a cluster multiset on fewer leaves, which by the induction hypothesis recovers two unique trees up to subtree swaps.

Furthermore, the subclusters of $A_1$ are able to be recovered as two subtrees $T_A$ and $T_B$ by Lemma \ref{l:split.clusterset} (as due to minimality of $A_1$, the trees $T_A$ and $T_B$ have no clusters in common), and either can be adjoined to either recovered quotient tree by identifying the root of $T_1$ or $T_2$ with $a_1$, recovering our original pair of trees up to one additional possible subtree swap. 

Finally, the fact that a pair of binary phylogenetic trees $T_1$ and $T_2$ is always recoverable from their $3$-bucket, up to a finite sequence of subtree swaps follows from the above proof, together with Corollary \ref{labelled_triples_same_clusters}.
\end{proof}

We will, in the next section, consider the same problem for the $(n-1)$-bucket. In this case, we again achieve recoverability on all but a finite set of pathological examples.

\section{Labelled trees and \texorpdfstring{$(n-1)$}{(n-1)}-buckets}\label{s:labeln-1}

We now consider the $(n-1)$-buckets of pairs of labelled trees. We first show that it is sufficient to find four distinct subtrees on $(n-1)$ leaves.

\begin{Theorem} \label{n-1_labelled_recover_from_4}
    Let $T$ be a labelled phylogenetic tree on $X$, where $|X|=n$, and let $T_a$, $T_b$, $T_c$, and $T_d$ be distinct subtrees of $T$ with $n-1$ leaves. Then $T$ is the unique tree on $X$ with $T_a$, $T_b$, $T_c$, and $T_d$ as subtrees.
\end{Theorem}

\begin{proof}
    Consider some $x, y, z \in X$. Then one of $T_a$, $T_b$, $T_c$ and $T_d$ contain $x$, $y$, and $z$ as leaves. Let this tree be $T'$. Consider the rooted triple of $x$, $y$, $z$ in $T'$. As $T'$ is a subtree of $T$, this rooted triple is also a rooted triple of $T$. As this holds for any choice of $x$, $y$, and $z$, any tree with $T_a$, $T_b$, $T_c$, and $T_d$ as subtrees must have the same set of rooted triples as $T$. But by Theorem \ref{t:3-recoverable}, $T$ is the unique tree with its rooted triples. Thus $T$ is the unique tree on $X$ with $T_a$, $T_b$, $T_c$, and $T_d$ as subtrees.
\end{proof}

If one can find these four distinct trees, then recovery is possible. We will show that this is always possible for trees on $n > 6$ leaves.

\begin{Theorem}
    Let $T_1$ and $T_2$ be labelled phylogenetic trees on $X$. If $|X| = n > 6$, then $T_1$ and $T_2$ are recoverable from their $(n-1)$-bucket.
\end{Theorem}

\begin{proof}
    By Theorem \ref{smaller_bucket_from_larger}, we can determine $B_3(\{T_1, T_2\})$ from $B_{n-1}(\{T_1, T_2\})$. Thus by Corollary \ref{labelled_triples_same_cherries}, we can determine the multiset of cherries of $T_1$ and $T_2$.
    
    Suppose there is some cherry $(a, b)$ in this multiset with multiplicity 2. Then $(a, b)$ is a cherry of both $T_1$ and $T_2$. Consider the subtrees in $B_{n-1}(\{T_1, T_2\})$ which do not have $a$ is a leaf. As we know that $(a, b)$ is a cherry in both $T_1$ and $T_2$ if we modify these subtrees by subdividing the edge leading to $b$, and attaching $a$ as a child of this new vertex we must obtain $T_1$ and $T_2$, thus $T_1$ and $T_2$ are recoverable from their $(n-1)$-bucket.
    
    Now suppose that no cherries in the multiset of cherries of $T_1$ and $T_2$ have multiplicity 2. Then there must be some cherry $(a, b)$ which is a cherry of exactly one of $T_1$ and $T_2$. Say without loss of generality that $(a, b)$ is a cherry of $T_1$. Then any subtree in $B_{n-1}(\{T_1, T_2\})$ which has both $a$ and $b$ as leaves, but not a cherry must be a subtree of $T_2$. Suppose there is some leaf $c \in X$ such that the subtree of $T_2$ obtained by removing $c$ has $(a, b)$ as a cherry. Then either $(a, c)$ or $(b, c)$ is a cherry of $T_2$, and so $(a, b)$ is not a cherry of any other subtree of $T_2$ with $n-1$ leaves, and so, as $T_2$ has at least 7 leaves, $B_{n-1}(\{T_1, T_2\})$ contains at least $7-3=4$ trees which contain both $a$ and $b$ as leaves, but not a cherry. If no such $c$ exists then $B_{n-1}(\{T_1, T_2\})$ contains at least $4$ trees which contain both $a$ and $b$ as leaves, but not a cherry. Thus there are at least $4$ trees in $B_{n-1}(\{T_1, T_2\})$ which must be subtrees of only $T_2$, and so by Theorem \ref{n-1_labelled_recover_from_4}, $T_2$ is recoverable from $B_{n-1}(\{T_1, T_2\})$. If we then consider the $(n-1)$-multideck of $T_2$, we can determine the $(n-1)$-multideck of $T_1$, which has at least 4 elements as $n>6 \geq 4$ and so by Theorem \ref{n-1_labelled_recover_from_4} $T_1$ is also recoverable, and so $T_1$ and $T_2$ are recoverable from their $(n-1)$-bucket.
\end{proof}

Finally, we consider the remaining possible trees on $n \le 6$ leaves.

\begin{Theorem}
    Let $T_1$ and $T_2$ be labelled phylogenetic trees on $X$. If $n \le 4$, $T_1$ and $T_2$ are recoverable from their $(n-1)$-bucket. Furthermore, if $n=5$ or $6$, $T_1$ and $T_2$ are recoverable from their $(n-1)$-bucket unless they are the pathological examples depicted in Figures \ref{fig:five_leaf_labelled_counterexample_1}, \ref{fig:five_leaf_labelled_counterexample_2} and \ref{fig:six_leaf_labelled_counterexample}, up to relabelling of the leaves.
\end{Theorem}

\begin{proof}
    If $n \le 3$, there is only one tree on $n$ leaves, so all possible pairs of trees are recoverable.  

    For $n=4$, by Theorem \ref{t:full_subtree_swap}, the only possible way to generate two distinct pairs of trees with the same $(n-1)$-bucket is via subtree swaps. However, subtree swaps of subtrees on $1$ or $2$ leaves do not result in a distinct pair of trees (because there is only one possible tree on each of one leaf and two leaves), and if a $3$-leaf subtree is swapped, the fourth leaf must be a leaf child of the root and the unique remaining fourth leaf, and hence any subtree swap results in the same pair of trees in reverse order.

    For $n=5$ and $n=6$, we performed a computer-assisted search, for which our code is provided as supplementary data. Up to relabelling of the leaves there exist $2$ pairs of pairs of trees on $5$ leaves which cannot be recovered from their $4$-bucket, shown in Figures \ref{fig:five_leaf_labelled_counterexample_1} and \ref{fig:five_leaf_labelled_counterexample_2} (to be explicit, by a relabelling of the leaves we mean we consider two pairs of trees on $X$ to be equivalent if there exists a permutation of the leaf labels which transforms one pair into the other). Up to relabelling of the leaves there exists a single pair of pairs of trees on $6$ leaves which cannot be recovered from their $5$-bucket, shown in Figure \ref{fig:six_leaf_labelled_counterexample}.
\end{proof}

\begin{figure}[ht]
    \centering
    \begin{subfigure}[b]{0.4\textwidth}
        \caption*{$A_1$}
        \includegraphics[width=\textwidth]{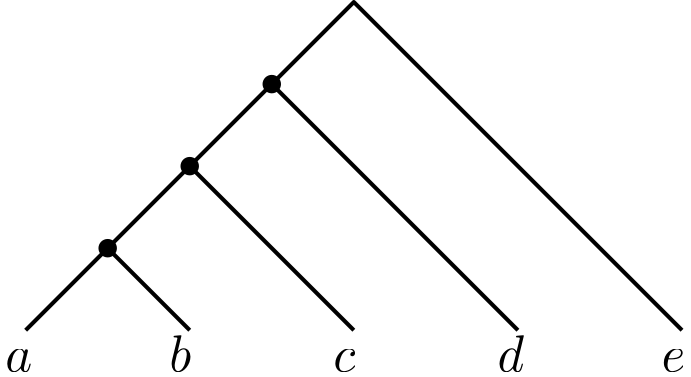}
    \end{subfigure}
    \hfill
    \begin{subfigure}[b]{0.4\textwidth}
        \caption*{$A_2$}
        \includegraphics[width=\textwidth]{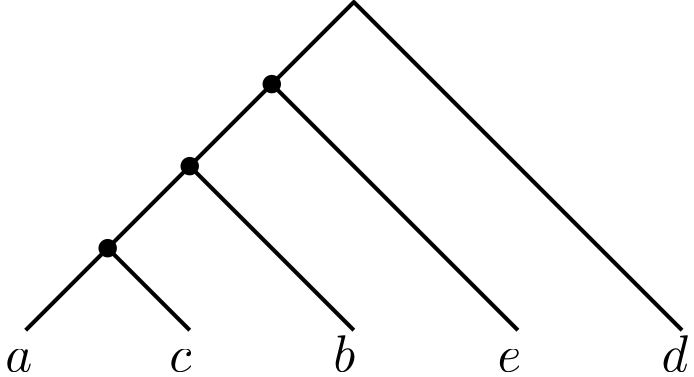}
    \end{subfigure}
    \newline\newline
    \begin{subfigure}[b]{0.4\textwidth}
        \caption*{$B_1$}
        \includegraphics[width=\textwidth]{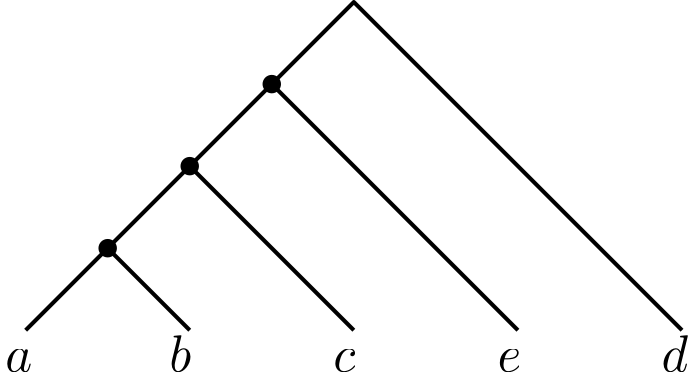}
    \end{subfigure}
    \hfill
    \begin{subfigure}[b]{0.4\textwidth}
        \caption*{$B_2$}
        \includegraphics[width=\textwidth]{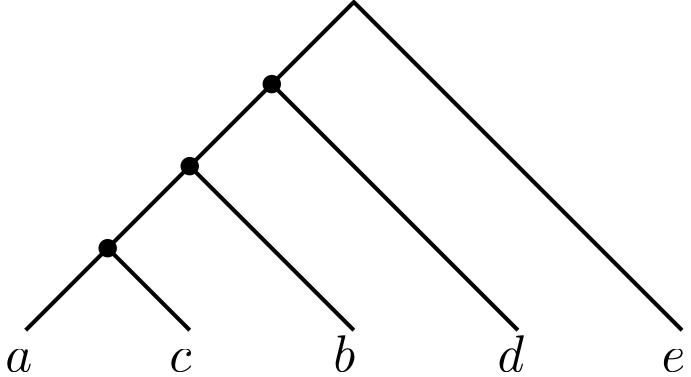}
    \end{subfigure}    
    \caption{The pairs $\{A_1, A_2\}$ and $\{B_1, B_2\}$ of labelled trees on $5$ leaves which have identical $4$-buckets.}
    \label{fig:five_leaf_labelled_counterexample_1}
\end{figure}

\begin{figure}[ht]
    \centering
    \begin{subfigure}[b]{0.4\textwidth}
        \caption*{$C_1$}
        \includegraphics[width=\textwidth]{5_leaf_labelled_counterexample_1_1.png}
    \end{subfigure}
    \hfill
    \begin{subfigure}[b]{0.4\textwidth}
        \caption*{$C_2$}
        \includegraphics[width=\textwidth]{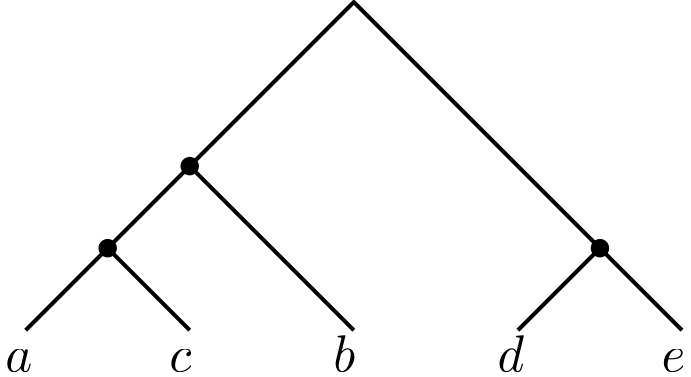}
    \end{subfigure}
    \newline
    \newline
    \begin{subfigure}[b]{0.4\textwidth}
        \caption*{$D_1$}
        \includegraphics[width=\textwidth]{5_leaf_labelled_counterexample_1_4.png}
    \end{subfigure}
    \hfill
    \begin{subfigure}[b]{0.4\textwidth}
        \caption*{$D_2$}
        \includegraphics[width=\textwidth]{5_leaf_labelled_counterexample_2_4.png}
    \end{subfigure}    
    \caption{The pairs $\{C_1, C_2\}$ and $\{D_1, D_2\}$ of labelled trees on $5$ leaves which have identical $4$-buckets.}
    \label{fig:five_leaf_labelled_counterexample_2}
\end{figure}

\begin{figure}[ht]
    \centering
    \begin{subfigure}[b]{0.4\textwidth}
        \caption*{$E_1$}
        \includegraphics[width=\textwidth]{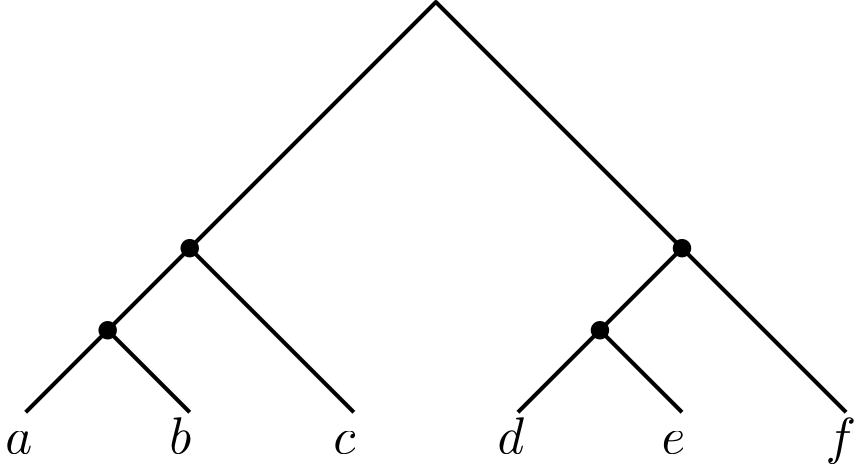}
    \end{subfigure}
    \hfill
    \begin{subfigure}[b]{0.4\textwidth}
        \caption*{$E_2$}
        \includegraphics[width=\textwidth]{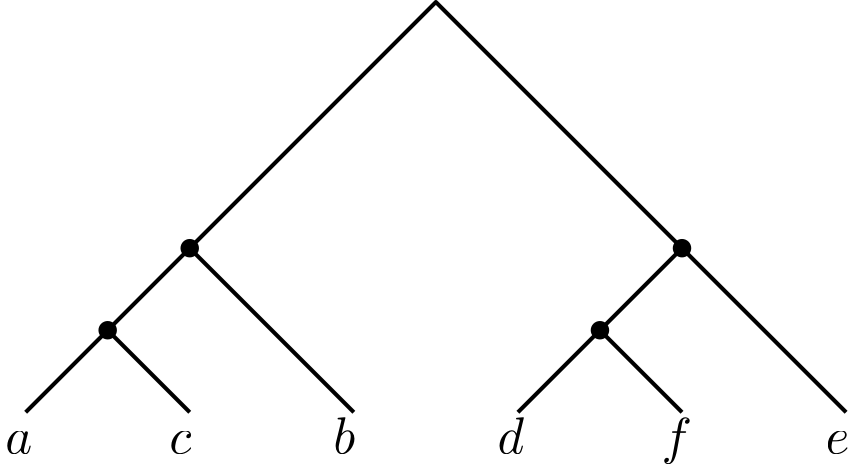}
    \end{subfigure}
    \newline\newline
    \begin{subfigure}[b]{0.4\textwidth}
        \caption*{$F_1$}
        \includegraphics[width=\textwidth]{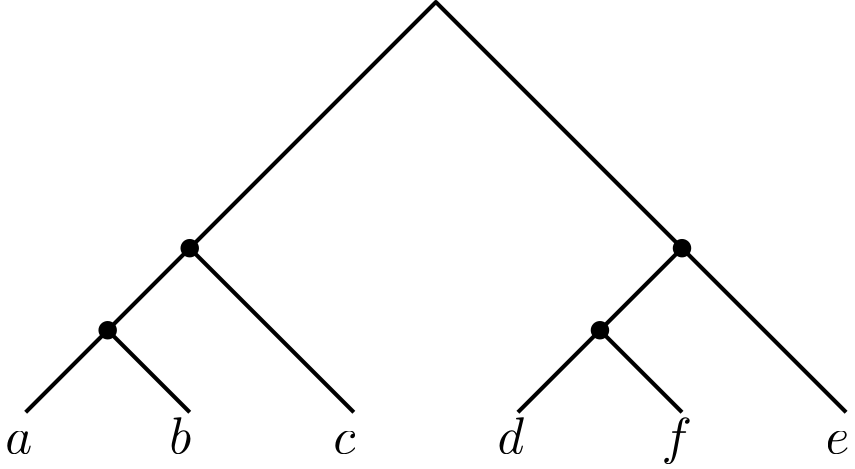}
    \end{subfigure}
    \hfill
    \begin{subfigure}[b]{0.4\textwidth}
        \caption*{$F_2$}
        \includegraphics[width=\textwidth]{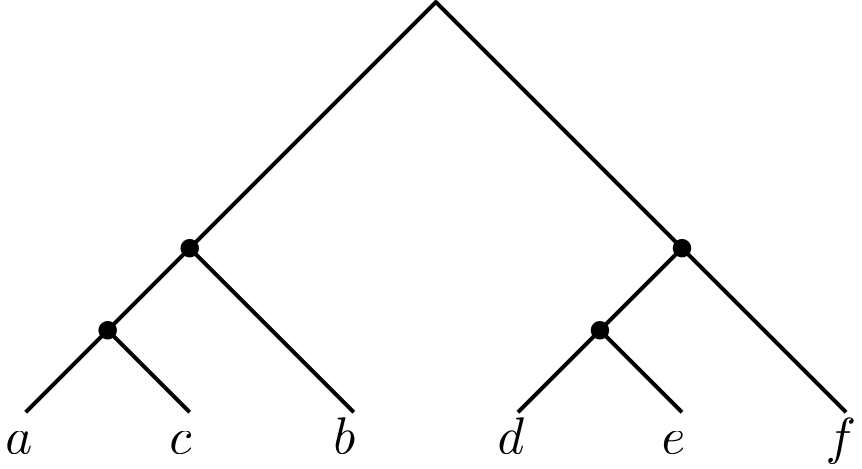}
    \end{subfigure}    
    \caption{The pairs $\{E_1, E_2\}$ and $\{F_1, F_2\}$ of labelled trees on $6$ leaves which have identical $5$-buckets.}
    \label{fig:six_leaf_labelled_counterexample}
\end{figure}

We have thus again shown recoverability for $(n-1)$-buckets of two phylogenetic trees except for a finite number of pathological examples (but this time, in the labelled case).

\clearpage

\section{Discussion}\label{s:discuss}

Our results suggest that reconstruction from mixed phylogenetic signals behaves fundamentally differently from classical phylogenetic reconstruction. While reconstruction from a single tree is unique from the set of its subtrees with $k$ leaves for labelled trees, and with $n-1$ leaves for unlabelled trees, $k$-buckets admit both finite pathological coincidences and infinite families of ambiguities arising from subtree swaps.

More specifically, we have demonstrated that it is always possible to recover these trees from their $(n-1)$-buckets in both the labelled and unlabelled cases, aside from some small pathological counterexamples (restricted to $n<9$ in the unlabelled case and $n<7$ in the labelled case). We have also characterised exactly when a pair of trees cannot be recovered from their $3$-bucket. Finally, we have also considered a number of statistics that can be shown to be recoverable - either always, in the case of root balance and pendant depth, or with sufficiently different input trees, in the case of height.

We anticipate that future research in this area will likely fall under one of three general categories, any of which may prove interesting. 

The first is considering larger tree sets or subtrees with different numbers of leaves. For instance, we have only considered cases with buckets sourced from two labelled trees, and subtrees with $(n-1)$ or $3$ leaves. We have shown that there exist an infinite number of unrecoverable pairs for $3$-buckets on two labelled trees, but a finite number for $(n-1)$-buckets. Discovering exactly when this transition happens would be an intriguing direction. Similarly, in the unlabelled case, we saw that recoverability is completely impossible from the $3$-bucket, but outside this phenomenon, is there $3<k<n-1$ for which recoverability is sometimes possible, but with an infinite family of counterexamples? If so, at what point does this transition occur?

A further natural question, of course, is how the behaviour changes as we modify the number of trees, which we shall denote $m$. Are the characterisations notably different if $m=3$? If this proves tractable, an interesting direction would be asymptotics, as $m$ approaches infinity. It seems likely that for fixed $n$, as $m$ increases recoverability becomes more difficult, but does the proportion of recoverable trees become vanishingly small, or approach a fixed proportion? 

Similar to the question of the phase transition from infinite to finite counterexamples when modifying $k$, if our supposition that increasing $m$ makes recoverability increasingly difficult, one can also ask at what point does recoverability become less likely than unrecoverability? Is there an $m$ at which the finite family of counterexamples for $(n-1)$-buckets becomes an infinite family?

The second promising direction is that of partial information, such as attempting to recover input trees when not all of the subtrees in the bucket are provided (or indeed, if some small number of them are incorrect). This may be done, for instance, by restricting the number to some proportion of $n$, or could also be considered from the perspective of decks instead of multidecks, in which we are not given multiplicities of each subtree (as in \cite{Decks_of_rooted_binary_trees}). For instance, our proof of recoverability for the labelled case of the $(n-1)$-bucket requires leveraging four specific trees in the $(n-1)$-bucket, but the problem may become considerably more difficult if we cannot guarantee that they are contained in our bucket. In particular, it would be striking if one could show a minimum required number of subtrees required to recover our main trees. Such results could provide guidance for biologists regarding the amount of partial information required before attempting reconstruction from mixed phylogenetic signals

Finally, the third possible avenue for research we have identified is of investigating the properties that have already been introduced (although potentially also in the contexts outlined in the other directions above). For instance, what is the algorithmic complexity of recovering the set of input trees given a $k$-bucket in general? Secondly, we did not characterise the image of $k$-buckets in the present manuscript, and any future research would certainly benefit if there was a characterisation beyond something as simple as ``these subtrees can be partitioned into $m$ sets of compatible subtrees''.

\section*{Acknowledgements} 
The authors wish to thank Andrew Francis, Thomas Britz and Kaitlyn O'Reilly for helpful discussions throughout the development of this manuscript. The authors would further like to thank Yukihiro Murakami and Kristina Wicke for helpful initial conversations at the inception of this project, 
at the International Workshop on Probabilistic, Combinatorial, and Algorithmic Phylogenetics, Taipei, 2026. SB's research was supported by the Commonwealth through an Australian Government Research Training Program Scholarship [DOI: \url{https://doi.org/10.82133/C42F-K220}]. MH's research was supported by the Australian Government through the Australian Research Council's Discovery Projects funding scheme (project DP260102678). The views expressed herein are those of the authors and are not necessarily those of the Australian Government or Australian Research Council.

\section*{Conflict of interest} The authors herewith certify that they have no affiliations with or involvement in any
organisation or entity with any financial (such as honoraria; educational grants; participation in speakers’ bureaus;
membership, employment, consultancies, stock ownership, or other equity interest; and expert testimony or patent-licensing
arrangements) or non-financial (such as personal or professional relationships, affiliations, knowledge or beliefs) interest in the subject matter discussed in this manuscript.

\section*{Data availability statement} 
Data sharing is not applicable to this article as no new data were created or analysed in this study.

\bibliographystyle{alpha}
\bibliography{sample}

\end{document}